%% file: sigconf.tex
\documentclass[sigconf,authorversion,nonacm]{acmart} 
\AtBeginDocument{%
  }

\usepackage{nicefrac}       
\usepackage{amsmath}

\newtheorem{definition}{Definition}
\newtheorem{lemma}{Lemma}
\newtheorem{theorem}{Theorem}
\newtheorem{proposition}{Proposition}

\usepackage{amsmath}
\usepackage{amsfonts}
\usepackage{xcolor}
\usepackage{mathtools}
\usepackage{algorithm}
\usepackage{algpseudocode}

\usepackage{adjustbox}
\usepackage{booktabs}
\usepackage{multirow}
\usepackage{tabularx}
\usepackage{xcolor}
\usepackage{colortbl}
\usepackage{amsmath}
\newcommand{\Romannum}[1]{%
  \uppercase\expandafter{\romannumeral #1}%
}

\usepackage{xcolor}

\makeatletter
\@ifundefined{proof}{%
  \providecommand{\qedsymbol}{$\square$}
  \newenvironment{proof}[1][Proof]{%
    \par\noindent\textit{#1.}\ %
  }{%
    \hfill\qedsymbol\par
  }
}{}
\makeatother

\definecolor{adred}{HTML}{F07D70}
\definecolor{adblue}{HTML}{5CADD5}
\definecolor{adgreen}{HTML}{6CC27F}
\definecolor{winnerbg}{RGB}{251,224,213}

\usepackage{soul}
\soulregister\texttt7

\setul{0.45ex}{3pt}

\DeclareRobustCommand{\redul}[1]{%
  \begingroup
  \setul{0.45ex}{3pt}%
  \setulcolor{adred}%
  \ul{#1}%
  \endgroup
}

\DeclareRobustCommand{\blueul}[1]{%
  \begingroup
  \setul{0.45ex}{3pt}%
  \setulcolor{adblue}%
  \ul{#1}%
  \endgroup
}

\DeclareRobustCommand{\greenul}[1]{%
  \begingroup
  \setul{0.45ex}{3pt}%
  \setulcolor{adgreen}%
  \ul{#1}%
  \endgroup
}

\makeatletter
\DeclareRobustCommand{\titlecolorul}[2]{%
  \begingroup
  \setbox\z@=\hbox{#2}%
  \leavevmode
  \rlap{%
    \textcolor{#1}{%
      \rule[\dimexpr-.45ex-3pt\relax]{\wd\z@}{3pt}%
    }%
  }%
  \unhbox\z@
  \endgroup
}
\DeclareRobustCommand{\titlemarkup}[1]{%
  \begingroup
  \def\redul##1{\titlecolorul{adred}{##1}}%
  \def\blueul##1{\titlecolorul{adblue}{##1}}%
  \def\greenul##1{\titlecolorul{adgreen}{##1}}%
  #1%
  \endgroup
}
\makeatother

\makeatletter
\newdimen\SOUL@winneruldp
\newdimen\SOUL@winnerulht

\def\SOUL@winnerulleaders{%
  \leaders\hrule
    \@depth\SOUL@winneruldp
    \@height\SOUL@winnerulht\relax
}
\def\SOUL@winnerdecorate#1{{%
  \setbox\z@\hbox{#1}%
  \SOUL@dimen=\wd\z@
  \SOUL@dimeni=\SOUL@uloverlap
  \advance\SOUL@dimen2\SOUL@dimeni
  \rlap{%
    \null
    \kern-\SOUL@dimeni
    \textcolor{winnerbg}{%
      \rule[-.75ex]{\SOUL@dimen}{2.5ex}%
    }%
  }%
  \rlap{%
    \null
    \kern-\SOUL@dimeni
    \SOUL@winnerulcolor{%
      \SOUL@winnerulleaders\hskip\SOUL@dimen\kern\z@
    }%
  }%
  \unhcopy\z@
}}
\def\SOUL@winnerpreamble{%
  \SOUL@winneruldp=.45ex
  \SOUL@winnerulht=-.45ex
  \advance\SOUL@winneruldp3pt
  \spaceskip\SOUL@spaceskip
}
\def\SOUL@winnereverysyllable{%
  \SOUL@winnerdecorate{%
    \the\SOUL@syllable
    \SOUL@setkern\SOUL@charkern
  }%
}
\def\SOUL@winnereveryspace#1{%
  #1%
  \SOUL@dimen=\spaceskip
  \rlap{%
    \textcolor{winnerbg}{%
      \rule[-.75ex]{\SOUL@dimen}{2.5ex}%
    }%
  }%
  \SOUL@winnerulcolor{%
    \SOUL@winnerulleaders\hskip\spaceskip\kern\z@
  }%
  \null
}
\def\SOUL@winnereveryhyphen{%
  \discretionary{%
    \unkern
    \SOUL@winnerdecorate{%
      \SOUL@setkern\SOUL@hyphkern
      \SOUL@sethyphenchar
    }%
  }{}{}%
}
\def\SOUL@winnereveryexhyphen#1{%
  \SOUL@setkern\SOUL@hyphkern
  \SOUL@winnerdecorate{#1}%
  \discretionary{}{}{%
    \SOUL@setkern\SOUL@charkern
  }%
}
\def\SOUL@winnerulsetup#1{%
  \SOUL@setup
  \def\SOUL@winnerulcolor{\textcolor{#1}}%
  \let\SOUL@preamble\SOUL@winnerpreamble
  \let\SOUL@everysyllable\SOUL@winnereverysyllable
  \let\SOUL@everyspace\SOUL@winnereveryspace
  \let\SOUL@everyhyphen\SOUL@winnereveryhyphen
  \let\SOUL@everyexhyphen\SOUL@winnereveryexhyphen
}
\DeclareRobustCommand*\winnerredul{\SOUL@winnerulsetup{adred}\SOUL@}
\DeclareRobustCommand*\winnerblueul{\SOUL@winnerulsetup{adblue}\SOUL@}
\DeclareRobustCommand*\winnergreenul{\SOUL@winnerulsetup{adgreen}\SOUL@}
\makeatother

\DeclareRobustCommand{\winnerhl}[1]{%
  \begingroup
  \sethlcolor{winnerbg}%
  \hl{#1}%
  \endgroup
}

\soulregister\redul1
\soulregister\blueul1
\soulregister\greenul1

\usepackage{tikz}
\usetikzlibrary{arrows.meta, shapes.geometric, positioning, calc, decorations.pathmorphing, patterns, backgrounds}

\usepackage{tcolorbox}
\usepackage{tabularx}
\usepackage{xcolor}

\usepackage{framed}

\colorlet{shadecolor}{gray!10}
\colorlet{TFFrameColor}{gray!50}

\usepackage[framemethod=TikZ]{mdframed}

\newmdenv[
  topline=false,
  bottomline=false,
  rightline=false,
  leftline=true,
  linecolor=gray!60,
  linewidth=2pt,
  skipabove=6pt,
  skipbelow=2pt,
  innerleftmargin=6pt,
  innerrightmargin=0pt,
  innertopmargin=6pt,
  innerbottommargin=6pt
]{paperquote}

\setcopyright{none}
\acmConference[KDD '27]
  {ACM SIGKDD Conference on Knowledge Discovery and Data Mining}
  {August 1--5, 2027}
  {San Jose, CA, USA}

\begin{document}

\title[Token-Level Advertising]{\titlemarkup{\redul{Token}\blueul{-Level}\greenul{ Advertising}}}\titlenote{Preprint. Work in progress. Underline colors indicate the advertiser source of each token (Qwen3 tokenizer is used throughout the paper). In our mechanism, advertisers can influence the final response at token-level granularity.}

\author{Hanbing Liu}
\orcid{0009-0000-4582-3340}
\affiliation{%
  \institution{Renmin University of China}
  \city{Beijing}
  \country{China}
}
\email{liuhanbing@ruc.edu.cn}

\author{Bowei Zhang}
\affiliation{%
  \institution{Renmin University of China}
  \city{Beijing}
  \country{China}
}
\email{2023201812@ruc.edu.cn}

\author{Changyuan Yu}
\affiliation{%
  \city{Beijing}
  \country{China}
}
\email{changyuanyu83@outlook.com}

\author{Yinyu Ye}
\affiliation{%
  \institution{Stanford University}
  \city{Stanford}
  \state{CA}
  \country{USA}
}
\email{yinyu-ye@stanford.edu}

\author{Qi Qi}
\correspondingauthor
\affiliation{%
  \institution{Renmin University of China}
  \city{Beijing}
  \country{China}
}
\email{qi.qi@ruc.edu.cn}


\begin{abstract}
Generative AI is transforming how people access information, challenging traditional advertising mechanisms built around predefined slots. Towards generation-native advertising, we propose the Latent Advertiser Mixture Auction (LAMA), a token-level advertising mechanism that embeds advertiser influence directly into the generation process. Advertisers report local continuation values that induce advertiser-specific next-token policies, from which the platform decodes through a latent mixture while updating an allocation posterior. We show that LAMA satisfies Markov DSIC and IR, and achieves near-optimal KL-regularized welfare. We further develop a learning-based implementation that reconstructs the required reports online from learned local advantages and root values. Proof-of-concept experiments on real-world commercial-search query splits show that LAMA improves platform welfare and revenue while maintaining user-facing response quality, providing initial evidence for the feasibility of generation-native advertising.
\end{abstract}





\maketitle

\section{Introduction}\label{sec:intro}
\input{sections/intro}


\section{Model}\label{sec:model}
\input{sections/model}

\section{Latent Advertiser Mixture Auction}\label{sec:advertisement}
\input{sections/advertisement}

\section{Practical Implementation}\label{sec:implementation}
\input{sections/implementation}

\section{Experiments}\label{sec:exp}
\input{sections/exp}

\section{Related Work}\label{sec:related}

\input{sections/related}

\section{Conclusion}\label{sec:conclusion}
In this work, we introduced generation-native advertising, where advertising opportunities are shaped during generation rather than allocated only after they have been formed. LAMA instantiates this idea at token level, achieving Markov DSIC and IR with bounded welfare loss, together with a practical learning-based implementation. Proof-of-concept experiments show that LAMA can improve monetization while maintaining user-facing response quality. More broadly, our results point toward a new direction for advertising mechanism design in generative interfaces, where mechanisms govern not only who receives an advertising opportunity, but also how that opportunity emerges through generation.

\bibliographystyle{ACM-Reference-Format}
\bibliography{llm_advertising_papers,classic_ad_auction_papers}

\appendix

\section{Report-Model Training}
\label{app:practical-training}
\input{sections/training_algorithm}

\section{Potential Business Format}
\label{app:potential-business-format}
\input{sections/potential_business_format}



\section{Proofs}\label{sec:proofs}
\input{sections/proof}

\section{Additional Experiment Results}
\label{app:vertical-advertiser-overview}
\input{sections/vertical_advertiser_overview}

\end{document}

%% file: sections/intro.tex
\begin{figure}[t!]
    \centering
    \includegraphics[width=\columnwidth]{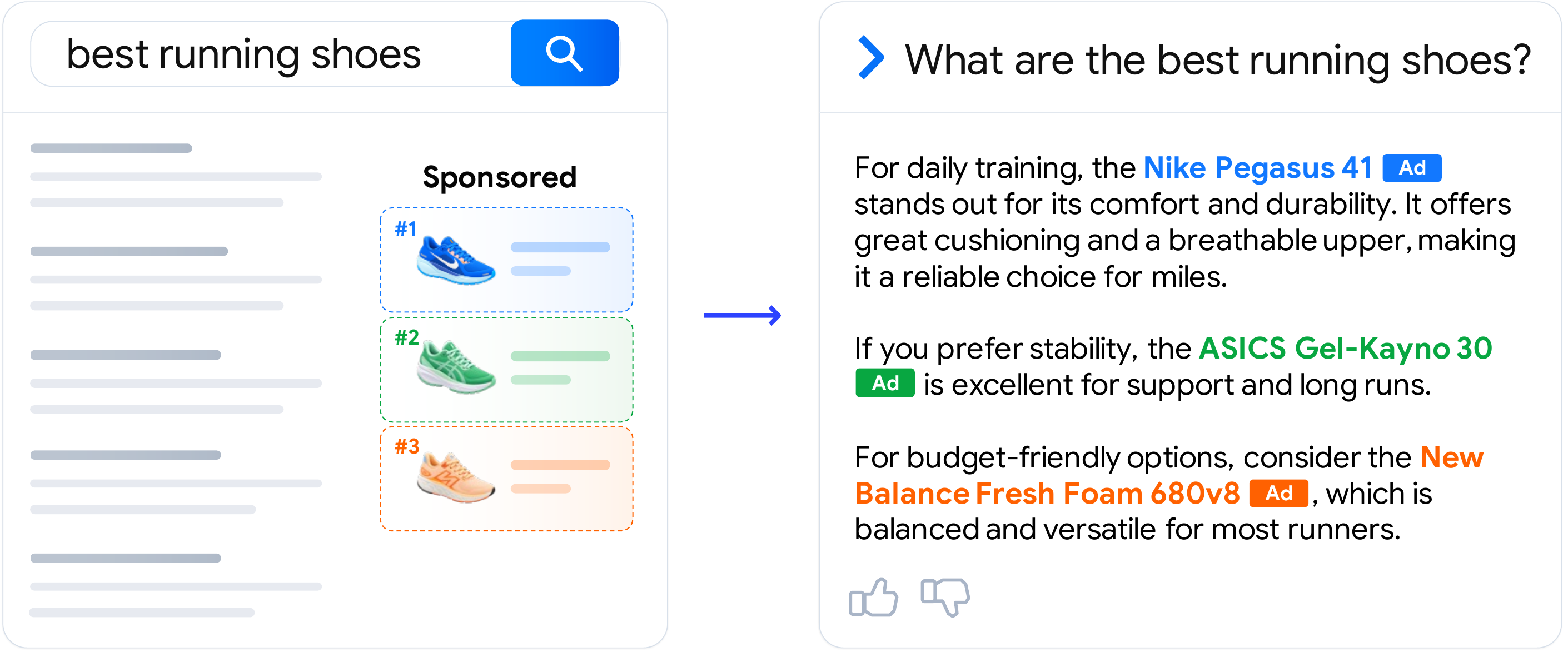}
    \caption{In generation-native advertising, ad opportunities are no longer predefined slots: when they emerge, which advertisers are shown, and how they are presented jointly shape user experience and marketing effectiveness.
}
    \label{fig:generation-native-advertising}
\end{figure}

\begin{paperquote}
\emph{``The medium is the message.''}
\par\smallskip
\hfill{--- Marshall McLuhan, \textit{Understanding Media} (1964)}
\nocite{mcluhan1964understanding}
\end{paperquote}

The medium through which people access online information is undergoing a fundamental transformation. For decades, internet interfaces have largely organized information by retrieving, ranking, and presenting content that already exists. Search engines return ranked results; recommendation systems select items from a collection. Generative AI introduces a different mode of information access. Rather than simply choosing what to present, AI-native interfaces synthesize responses dynamically around user intent, with the content itself taking shape as the interaction unfolds. The medium is no longer merely organizing information; it is generating it.

Advertising has historically been built around the structure of the medium in which it appears. In traditional internet interfaces, this has led to a simple but powerful abstraction: the slot.
Whether it is a sponsored-search position, a reserved display placement, or a promoted position in a recommendation feed, the commercial opportunity can typically be defined in advance and then assigned to advertisers.
But this abstraction becomes less natural when the surrounding content is itself generated. 
As shown in Figure~\ref{fig:generation-native-advertising}, a product may become relevant only as an answer develops in a particular direction, and an early generation choice can make a later brand mention either natural or jarring. In such settings, generation does more than fill an advertising opportunity—it can help create one. The advertising opportunity is no longer merely a slot inside generated content; it can be shaped by the generation process itself.

This shift calls for a corresponding change in mechanism design. Traditional advertising mechanisms largely separate two questions: what opportunities exist, and who should receive them? Auction theory has developed powerful tools for the latter, typically taking the former as given. 
This separation is natural when the advertising slot exists before the auction begins, but begins to break down when the generation trajectory itself determines whether and how an advertiser can be meaningfully incorporated.
Thus, a mechanism that only allocates predefined opportunities leaves an important part of the generative medium outside the allocation process.

Recent work has begun to rethink what constitutes an advertising opportunity in generative interfaces. Rather than relying only on fixed slots, some approaches define opportunities over segments or context-dependent positions within generated text, while others first construct complete candidate responses and then select, aggregate, or insert sponsored content at the response level. These abstractions make advertising substantially more adaptive to generated context: where an ad can appear, or which sponsored response is appropriate, may now depend on the content produced by the model. Yet the mechanism typically operates on an opportunity that has already been identified—a segment, a position, or a completed response. Generation shapes the opportunity, but is not itself the allocation process.

In this paper, we take the next step and formulate a generation-native paradigm for computational advertising, in which advertising is directly embedded into the generation process. Advertisers are treated as strategic participants whose messages can influence how the response evolves, while the platform jointly determines generated content, advertising allocation, and payments. Within this broader paradigm, we study a particularly native instantiation at the token level: advertisers may influence the response at every generation step. A user query initializes a sequential process, and at each generated prefix, the platform combines advertiser messages with the reference language model and the current mechanism state to determine what token to generate next and how the advertising allocation should evolve. The goal is to incorporate advertiser incentives and platform monetization into generation while preserving the naturalness and usefulness of the generated content.

Making advertising a native part of generation introduces several technical challenges. Most importantly, incentive design becomes inherently sequential: a report submitted at one prefix can change the future generation trajectory, which in turn changes the advertiser's subsequent opportunities and eventual payoff. Truthful behavior must therefore remain optimal not only at the beginning of generation, but also after any reachable future states. 
Advertiser influence must also be balanced against the reference language model so that monetization does not come at the expense of response quality. 
At the same time, language model inference already incurs substantial serving latency, requiring the mechanism to introduce only minimal additional overhead. 
Finally, a tractable mechanism should still admit a meaningful efficiency guarantee relative to an ideal welfare-maximizing benchmark.

To address these challenges, we propose the Latent Advertiser Mixture Auction (LAMA) for the fundamental single-winner setting. At each prefix, advertisers' local reports induce advertiser-specific next-token policies, and LAMA generates next token from a latent mixture of these policies. The mixture weights evolve as a Bayesian allocation posterior as tokens are observed, linking the generation trajectory to the eventual advertising allocation. Once generation terminates, this posterior determines the final allocation, together with payments designed for the sequential setting. We show that LAMA satisfies Markov dominant-strategy incentive compatibility and individual rationality.
Moreover, it achieves near-optimal KL-regularized welfare, with the optimality gap vanishing as the regularization diminishes.

We further develop a practical implementation that avoids explicitly reasoning over the full token tree or requiring advertisers to compute complex reports themselves. The platform provides reporting as a learned service: a shared advertiser-conditioned model recovers local signals along the realized generation path and anchors them with an advertiser-specific root estimate, allowing the reports required by LAMA to be reconstructed online. This design turns the theoretical sequential mechanism into a lightweight serving-time procedure that can operate over a small candidate set of advertisers with modest additional language-model computation.

We evaluate LAMA in simulation on real-world query splits with model-predicted impression values. We compare against heuristic mechanisms that allocate advertising either before or after generation, as well as a representative response-level aggregation mechanism. Across the evaluated settings, LAMA achieves the strongest platform welfare and revenue while maintaining high advertiser value and user-side response quality. These results provide an initial proof of concept for the generation-native advertising paradigm: allowing advertisers to participate directly in the token-level generation process can improve monetization while preserving high-quality generated responses.

Our contributions are summarized as follows.

\begin{enumerate}
    \item \textbf{Generation-native advertising.} We develop a token-level formulation of generation-native advertising, where advertiser participation is embedded directly into the sequential generation process and jointly shapes the response and eventual allocation.
    \item \textbf{Token-level mechanism design.} We propose LAMA for the fundamental single-winner setting and establish Markov DSIC and IR, 
    together with near-optimal KL-regularized welfare, with an additive gap of at most \(\beta \log |\mathcal{N}|\).
    \item \textbf{Practical implementation.} We develop a learning-based implementation that reconstructs the reports required by LAMA online and supports efficient serving.
    \item \textbf{Proof-of-concept evaluation.} Experiments on real-world search query splits show that LAMA improves platform welfare and revenue while maintaining advertiser value and user-facing response quality.
\end{enumerate}

%% file: sections/model.tex
We model token-level advertising as a sequential mechanism embedded in the
generation process of a language model. A user query initializes a token-level
decision process. At each non-terminal prefix, advertisers may submit local
reports based on the publicly observed state, and the platform uses these
reports together with its privately observed history to choose the
next token distribution, the advertising allocation, and payments. The central
distinction from classical sponsored-search models is that the advertising
opportunity is not an exogenously given slot. Instead, it is endogenously shaped
by the generated trajectory itself.

\paragraph{Context space and reference generation.}
Let \(\mathcal V\) be the token vocabulary and let \(L\) be the maximum
generation length\footnote{Natually, We assume \(|\mathcal{V}| \geq 3\).}. The context space is
\(
    \mathcal S := \mathcal V^{\le L}
    =
    \bigcup_{t=0}^{L}\mathcal V^t,
\)
where \(\mathcal V^0=\{\emptyset\}\). For two contexts \(x_1,x_2\in\mathcal S\),
we write \([x_1,x_2]\) for their concatenation whenever the resulting sequence
has length at most \(L\). A user query \(q\in\mathcal S\) is drawn from a query
distribution \(\mathcal D\) and serves as the initial state.

The platform has an organic language-model policy
\(
    \pi_{\mathrm{ref}}(\cdot\mid s)\in\Delta_{++}(\mathcal V),
    \, s\in\mathcal S,
\)
which describes the next-token distribution in the absence of advertising
intervention\footnote{Thus, the reference policy has full support over the entire context space.}. We also write \(\pi_0\) for \(\pi_{\mathrm{ref}}\). A trajectory
from query \(q\) is
\(
    \tau=(s_0=q,a_0,s_1,a_1,\ldots,s_{T_\tau}),
    \,
    s_{t+1}=[s_t,a_t].
\)
Let \(\texttt{EOS}\in\mathcal V\) denote the end-of-sequence token. A state is
terminal if it contains \(\texttt{EOS}\) or reaches the maximum length. We
denote the terminal set by \(\mathcal S_T\), and write \(L(s)\subseteq
\mathcal S_T\) for the set of terminal descendants of a non-terminal state
\(s\). For a trajectory \(\tau\), its realized terminal state is denoted by
\(\ell(\tau)=s_{T_\tau}\).

\paragraph{Advertisers and private types.}
There are \(n \geq 2\) advertisers, indexed by
\(
    \mathcal N=[n]=\{1,\ldots,n\}.
\)
Each advertiser \(i\) has a private terminal reward function
\(
    r_i:\mathcal S_T\to\mathbb R_{\ge 0}.
\)
The value \(r_i(\ell)\) represents advertiser \(i\)'s value from receiving
the advertising opportunity when the generated response terminates at
\(\ell\). 
We write \(r=(r_1,\ldots,r_n)\) for the type profile and
\(r_{-i}\) for the profile excluding advertiser \(i\).

The allocation of the advertising opportunity is represented by a vector
\(
    \lambda\in\Lambda\subseteq[0,1]^n,
\)
where \(\Lambda\) is an exogenously given feasible allocation set. For example,
\(\Lambda\) may encode slot-capacity constraints or exposure constraints. This paper focuses on the fundamental single-winner setting:
the advertising opportunity is endogenously shaped by the generation process, but the eventual sponsored outcome is assigned to a single advertiser.
Thus, $\Lambda = \Delta(\mathcal{N})$. The
realized gross value of advertiser \(i\) on trajectory \(\tau\) is therefore
\(
    \lambda_i r_i(\ell(\tau)).
\)

\paragraph{Platform histories.} To describe a mechanism that operates during
generation, we introduce the history privately observed by the platform. A
platform history at time \(t\) is
\[
h_t=(s_0,m_0,a_0,\ldots,s_{t-1},m_{t-1},a_{t-1},s_t),
\]
where \(s_0=q\), \(a_t\in \mathcal{V}\) is the token generated at time \(t\), and
\(m_t=(m_{1,t},\ldots,m_{n,t})\) denotes the profile of messages submitted by
the advertisers at time \(t\). Advertisers observe only the current state
\(s_t\), which is public, whereas the platform observes the entire history. The
state evolves deterministically according to \(s_{t+1}=[s_t,a_t]\). Let
\(\mathcal H\) denote the set of non-terminal platform histories, and let
\(s(h)\) be the current state at the end of history \(h\). A terminal platform
history is denoted by
\(
h_T=(s_0,m_0,a_0,\ldots,s_{T-1},m_{T-1},a_{T-1},s_T),
\)
where \(s_T\in\mathcal S_T\). We write \(\tau(h_T)\) for the token trajectory
induced by the terminal history \(h_T\), and \(\ell(h_T)=s_T\) for its terminal
state.

\paragraph{Sequential mechanisms.}

A sequential mechanism for token-level advertising is a tuple
\[
\Gamma=\bigl(\{\mathcal{M}_i(h)\}_{i\in \mathcal{N},h\in\mathcal H},x,\lambda,p\bigr).
\]
For each non-terminal platform history \(h\), \(\mathcal{M}_i(h)\) is advertiser
\(i\)'s feasible message space after that history. Given \(h\in\mathcal H\)
and a joint message
\(m=(m_1,\ldots,m_n)\in \mathcal{M}(h):=\prod_{i\in\mathcal N}\mathcal{M}_i(h)\), the next-token rule is
\(
x(\cdot\mid h,m)\in\Delta(\mathcal{V}).
\)
The instantaneous payment rule is
\(
p_i(h,m,a)\in\mathbb R,
\)
which specifies advertiser \(i\)'s payment when token \(a\) is generated after
history \(h\) under report profile \(m\). The final allocation rule is
\(
\lambda(h_T)\in\Lambda.
\)
Thus, the feasible message spaces, token-generation rule,
payments, and final allocation may depend on the entire platform history\footnote{When
no ambiguity arises, we sometimes use the current state \(s(h)\) in place of
the platform history \(h\) for brevity and intuition.}. In
particular, history-dependent message spaces allow the platform to check a new
report against values previously recorded in its report ledger.

The timing is as follows.
\begin{enumerate}
    \item The user query \(q\sim\mathcal D\) initializes the process at
    \(s_0=q\).
    \item At each non-terminal state \(s_t\), each advertiser \(i\)
    submits a message \(m_{i,t}\in \mathcal{M}_i(h_t)\). Let
    \(m_t=(m_{1,t},\ldots,m_{n,t})\).
    \item The mechanism chooses a next-token distribution
    \(x(\cdot\mid h_t,m_t)\), draws
    \(a_t\sim x(\cdot\mid h_t,m_t)\), and charges 
    \(p_i(h_t,m_t,a_t)\) for each \(i\in \mathcal{N}\).
    \item The state updates deterministically as \(s_{t+1}=[s_t,a_t]\). If
    \(s_{t+1}\in\mathcal S_T\), the process terminates and the mechanism implements
    the feasible allocation
    \(
    \lambda(h_{T})\in\Lambda .
    \)
\end{enumerate}


Direct mechanisms are a special case as the platform can ask each advertiser to submit its entire reward function. However, such mechanisms are impractical for
online generation because full reporting may be slow and may reveal
private information about trajectories that are never realized. Our definition therefore emphasizes lightweight reports, and we will see that they suffice to achieve
strong performance for the platform.

\paragraph{Strategies and induced policies.}
Let \(\Theta_i\) denote advertiser \(i\)'s type space. A
strategy of advertiser \(i\) maps its type and each non-terminal public state
to a message:
\(
    \sigma_i:\Theta_i\times(\mathcal S\setminus\mathcal S_T)\to
    \bigcup_{h\in\mathcal H}\mathcal M_i(h).
\)
A strategy is feasible if
\(\sigma_i(r_i,s(h))\in\mathcal M_i(h)\) at every history at which it is used.
Given \(\sigma=(\sigma_1,\ldots,\sigma_n)\) and type profile \(r\), the induced
joint report at history \(h\) is
\(
    m^{\sigma,r}(h)
    :=
    \left(
        \sigma_1(r_1,s(h)),\ldots,\sigma_n(r_n,s(h))
    \right).
\)
Together with the mechanism \(\Gamma\), these reports determine the continuation
token policy
\(
    \pi^{\Gamma}_{\sigma,r}(a\mid h)
    :=
    x\left(a\mid h,m^{\sigma,r}(h)\right),
    \, a\in\mathcal V,
\)
and hence a distribution \(\mathbb P^{\Gamma}_{\sigma,r}(\cdot\mid h)\) over
terminal histories extending \(h\).

\paragraph{Utilities.}
Advertisers have standard quasi-linear utilities. 
At a terminal history \(h_T\), 
advertiser \(i\)'s realized utility is
\(
    u_i^\Gamma(h_T;r_i)
    =
    \lambda_i(h_T) r_i(\ell(h_T))
    -
    \sum_{t=0}^{T-1}
    p_i(h_t,m_t,a_t).
\)
For a non-terminal platform history \(h\), a type profile \(r\), a strategy \(\sigma_i\), and opponent strategies
\(\sigma_{-i}\), the continuation expected utility \(U_i^\Gamma
    (h;r_i,\sigma_i,\sigma_{-i};r_{-i})\) is
\[
\begin{aligned}
    \mathbb E_{h_T\sim
    \mathbb P^{\Gamma}_{\sigma,r}(\cdot\mid h)}
    \left[
        \lambda_i(h_T) r_i(\ell(h_T))
        -
        \sum_{t=t(h)}^{T-1}
        p_i(h_t,m_t,a_t)
    \right].
\end{aligned}
\]


\paragraph{Truthful reporting.}
The truthful report is specified by a mechanism-specific type-related reporting rule
\(
    \eta_i:\Theta_i\times(\mathcal S\setminus\mathcal S_T)\to
    \bigcup_{h\in\mathcal H}\mathcal M_i(h).
\)
For example, the sub-tree reporting rule \(\eta_i(r_i,s)=r_i|_{L(s)}\) and the value function reporting rule \(
    \eta_i(s;r_i)
    =
    (
        V^*_{r_i}([s,a])
    )_{a\in\mathcal V}.
\)
It defines the truthful strategy
\(
    \sigma_i^{\mathrm{tr}}(r_i,s)
    :=
    \eta_i(s;r_i).
\)

Because a sequential mechanism allows advertisers to report information
dynamically, incentive compatibility requires that an advertiser cannot wait
for a favorable generated prefix and then manipulate its future reports. Formally,

\begin{definition}[Markov DSIC and IR]
The mechanism
\(\Gamma\) is Markov dominant-strategy incentive compatible
(Markov DSIC) if, for every advertiser \(i\), type profile \(r\), feasible
opponent strategy profile \(\sigma_{-i}\), reachable platform history \(h\) after
advertiser \(i\) has reported truthfully along the prefix, and feasible continuation deviation \(\sigma_i'\),
\(
    U_i^\Gamma
    \left(
        h;r_i,\sigma_i^{\mathrm{tr}},\sigma_{-i};r_{-i}
    \right)
    \ge
    U_i^\Gamma
    \left(
        h;r_i,\sigma_i',\sigma_{-i};r_{-i}
    \right).
\)
It is individually rational (IR) if, 
 truthful participation from the
initial state \(s_0=q\) gives nonnegative expected utility:
\(
    U_i^\Gamma
    \left(
        q;r_i,\sigma_i^{\mathrm{tr}},\sigma_{-i};r_{-i}
    \right)
    \ge 0.
\)
\end{definition}

\paragraph{KL-regularized social welfare.}
The platform's objective is to generate native and useful
content while allocating advertising opportunities efficiently. For a 
query \(q\), type profile \(r\), generation policy \(\pi\),
define the KL-regularized social welfare \(\mathrm{SW}_{\beta}(\pi,\lambda\mid q,r)\) as
\[
\begin{aligned}
    \mathbb E_{\tau\sim\pi(\cdot\mid q)}
    \left[
        \sum_{i \in \mathcal{N}}
        \lambda_i(h_T) r_i(\ell(\tau))
        -
        \beta
        \sum_{t=0}^{T_\tau-1}
        \log
        \frac{\pi(a_t\mid s_t)}
        {\pi_{\mathrm{ref}}(a_t\mid s_t)}
    \right].
\end{aligned}
\]
The first term is gross advertiser value under the generated content and realized allocation. The divergence
term penalizes deviation from the organic policy, which optimizes user experience.
\footnote{This KL-divergence penalty is standard in the reinforcement learning with human feedback (RLHF, \cite{ouyang2022training}) literature.}
Thus
\(\beta>0\) controls the balance between monetization and content naturalness.
When the query is also random, the expected welfare is
\(
    \mathbb E_{q\sim\mathcal D}
    \left[
        \mathrm{SW}_{\beta}(\pi,\lambda\mid q,r)
    \right].
\)

\paragraph{Mechanism design problem.}
For each query-type pair \((q,r)\), a sequential mechanism \(\Gamma\) and its truthful reporting rule induce a token policy \(\pi_{q,r}^\Gamma\)
and allocation rule \(\lambda_{q,r}^\Gamma\) from the initial state.
The platform chooses the mechanism to maximize induced welfare:
\[
\begin{aligned}
    \max_{\Gamma}
    \quad &
    \mathbb E_{q\sim\mathcal D,\,r\sim F(\cdot\mid q)}
    \left[
        \mathrm{SW}_{\beta}
        \left(
            \pi_{q,r}^\Gamma,\lambda_{q,r}^\Gamma
            \mid q,r
        \right)
    \right]
    \\
    \text{s.t.}
    \quad &
    \Gamma \text{ is Markov DSIC and IR},
    \\
    &
    \lambda_{q,r}^\Gamma(\tau)\in\Lambda,
    \,
    \forall q,r,\tau .
\end{aligned}
\]
Here \(F(\cdot\mid q)\) denotes the conditional distribution over the type
profiles of advertisers recalled for query \(q\).
For ex-post analysis, the same objective can be studied pointwise for each
realized \((q,r)\).

%% file: sections/advertisement.tex


We propose Latent Advertiser Mixture Auction (LAMA) to tackle this challenge.
The key idea is to view the generation process as Bayesian sequential
decision making that identifies the advertiser who values the generation
opportunity most. As we will see, this perspective leads to a simple algorithm
that is both fast and efficient.

\paragraph{Sequential reports.}
At a non-terminal state \(s\), advertiser \(i\) reports a child-value vector
\(
    \widehat V_i^s
    =
    (
        \widehat V_i^s([s,a])
    )_{a\in\mathcal V}.
\)
The intended truthful report is
\(
    \eta_i(s;r_i)
    =
    \left(
        V_i^*([s,a])
    \right)_{a\in\mathcal V},
\)
where \(V_i^*\) is the optimal soft value function induced by advertiser
\(i\)'s terminal reward \(r_i\):
\[
    V_i^*(s)
    =
    \max_{\pi}
    \mathbb E_{\tau\sim\pi(\cdot\mid s)}
    \left[
        r_i(\ell(\tau))
        -
        \beta
        \sum_{t=0}^{T_\tau-1}
        \log
        \frac{\pi(a_t\mid s_t)}
        {\pi_{\mathrm{ref}}(a_t\mid s_t)}
    \right].
\]
Given a reported child-value vector, define advertiser \(i\)'s induced local
policy by
\[
    \widehat\pi_i(a\mid s;\widehat V_i^s)
    :=
    \frac{
        \pi_{\mathrm{ref}}(a\mid s)
        \exp\left(
            \widehat V_i^s([s,a])/\beta
        \right)
    }{
        \sum_{a'\in\mathcal V}
        \pi_{\mathrm{ref}}(a'\mid s)
        \exp\left(
            \widehat V_i^s([s,a'])/\beta
        \right)
    }.
\]
This is the next-token policy that would be optimal if advertiser \(i\)'s
reported truthfully.

\paragraph{Allocation posterior and next-token rule.}
The mechanism maintains an allocation posterior
\(
    \rho(s)\in\Delta(\mathcal N).
\)
At the query \(q\), after receiving the first child-value reports, initialize
\(
    Z_i(q)
    :=
    \sum_{a\in\mathcal V}
    \pi_{\mathrm{ref}}(a\mid q)
    \exp(
        \widehat V_i^q([q,a])/\beta
    ),
\)
and set
\(
    \rho_i(q)
    :=
    \frac{Z_i(q)}
    {\sum_{j\in\mathcal N}Z_j(q)}.
\)

At a reached non-terminal state \(s\), after receiving reports
\(\widehat V^s=(\widehat V_1^s,\ldots,\widehat V_n^s)\), the mechanism samples hierarchically: first sample a latent advertiser
\(
    I\sim\rho(s)
\)
and then generate the next token from that advertiser's induced policy,
\(
    a\sim\widehat\pi_I(\cdot\mid s;\widehat V_I^s)
\).
Marginally, this sampling rule induces the mixture next-token
distribution
\[
    x(a\mid s,\widehat V^s,\rho(s))
    :=
    \sum_{i\in\mathcal N}
    \rho_i(s)
    \widehat\pi_i(a\mid s;\widehat V_i^s).
\]
The subsequent posterior update is precisely Bayesian inference about this
latent advertiser after observing the realized token.
After token \(a\) is realized, the posterior is updated by Bayes' rule:
\[
    \rho_i([s,a])
    =
    \frac{
        \rho_i(s)
        \widehat\pi_i(a\mid s;\widehat V_i^s)
    }{
        x(a\mid s,\widehat V^s,\rho(s))
    }.
\]
If a terminal history \(h_T\) is reached, the final allocation rule is
\(
    \lambda(h_T)=\rho(\ell(h_T))\in\Delta(\mathcal N).
\)

\paragraph{Payment rule.}
The payment rule has two components. First, after the initial reports at the
query prefix \(q\), advertiser \(i\) pays an entry fee \(p_i(q;\widehat V)\):
\[
    \rho_i(q)\widehat V_i(q)
    -
    \beta
    \log
    \sum_{j \in \mathcal{N}}
        \exp(
            \widehat V_j(q)/\beta
        )
    +
    \beta
    \log
    (
        1+\sum_{j\neq i}
        \exp(
            \widehat V_j(q)/\beta
        )
    ),
\]
where \(\widehat V_i(q)=\beta\log Z_i(q)\) is the root value induced by the
initial child-value report. Second, along the generated path, the mechanism
charges advertisers according to the change in their posterior-weighted reported
continuation value. At a transition \(s\to [s,a]\), define the
instantaneous payment
\[
    p_i(s,\widehat V^s,a)
    :=
    \rho_i([s,a])\widehat V_i^s([s,a])
    -
    \rho_i(s)\widehat V_i(s),
\]
where the parent value \(\widehat V_i(s)\) is the value previously recorded in
the platform ledger. Thus, the mechanism charges advertiser \(i\) when the
realized token makes its continuation prospects more favorable, and subsidizes
it when situation gets worse, preserving its
incentive to continue participating after an unfavorable transition.



\paragraph{Feasible message space.}
We require reports to satisfy the soft Bellman recursion.
This is easy to verify online: for each advertiser, the platform
only needs to store a scalar, i,e. the value previously reported for the current
state, and compare it with the soft Bellman aggregation of the newly reported
child values. Formally,

\begin{definition}[Bellman consistency]
\label{def:bellman-consistent-report}
At a history \(h\) with current state \(s=s(h)\), let \(\widehat V_i(s)\) be advertiser
\(i\)'s value for \(s\) already stored in the ledger. The feasible message
space is
\[
    \mathcal M_i(h)
    :=
    \begin{cases}
        \mathbb R^{|\mathcal V|},&s=q,\\
        \{v\in\mathbb R^{|\mathcal V|}:\beta\log\sum_{a\in\mathcal V}
    \pi_{\mathrm{ref}}(a\mid s)e^{v_a/\beta}=\widehat V_i(s)\},&s\ne q.
    \end{cases}
\]
Here \(v_a\) is the reported value of the child state \([s,a]\) in \(\widehat{V}^s_i\).
\end{definition}

We now ready to show the main incentive guarantee.

\begin{theorem}[Markov DSIC and IR of LAMA]
\label{thm:latent-advertiser-mixture-markov-dsic}
The LAMA with Bellman-consistent
message spaces is Markov DSIC and IR for the truthful reporting rule
\(
    \eta_i(s;r_i)
    =
    \left(
        V_i^*([s,a])
    \right)_{a\in\mathcal V}.
\)
\end{theorem}

\begin{proof}[Proof sketch]
First consider globally Bellman-consistent direct reports and write
\(z_i(s)=\exp(\widehat V_i(s)/\beta)\) and
\(\overline z(s)=\sum_j z_j(s)\). Bellman consistency makes each \(z_i\)
harmonic under \(\pi_{\mathrm{ref}}\). Substituting this identity into the
Bayesian update yields
\(
    \rho_i(s)=z_i(s)/\overline z(s)
\)
and
\(
    x(a\mid s)=\pi_{\mathrm{ref}}(a\mid s)
    \overline z([s,a])/\overline z(s)
\).
The factors telescope along a path, so the joint probability that terminal
state \(\ell\) is generated and allocated to advertiser \(i\) is
\[
    \alpha_{i,\ell}(\widehat V)
    =
    \frac{P_0(\ell\mid q)
    \exp(\widehat V_i(\ell)/\beta)}
    {\sum_{\ell'}P_0(\ell'\mid q)
    \sum_j\exp(\widehat V_j(\ell')/\beta)}.
\]
This is exactly the gradient, with respect to terminal reports, of the convex
log-sum-exp potential \(\Phi_q\) defined in the appendix. The telescoping LAMA
payment is the corresponding potential-based payment, so the convex supporting
hyperplane inequality gives direct truthfulness. Nonnegative rewards imply
\(z_i^*(q)\ge 1\), which makes truthful utility nonnegative relative to the
outside option.

It remains to connect online reports to this direct construction. Fix a reached
state \(s\) after a truthful prefix. Any feasible continuation deviation defines
a Bellman-consistent value function inside the subtree rooted at \(s\). When
\(s\ne q\), feasibility forces its root value to equal the truthful value
already stored at \(s\); when \(s=q\), there are no ancestors to preserve. Use
the deviating values inside the subtree and \(V_i^*\) everywhere else. The
result is a globally Bellman-consistent direct report with the same continuation
policy, allocation, and payments as the online deviation. Direct DSIC and IR
then immediately imply Markov DSIC and IR for the sequential mechanism.
\end{proof}


\paragraph{Efficiency.}
LAMA also provides strong efficiency guarantee. Intuitively, it implements two
soft lotteries: one over generated responses through KL-regularized decoding
and one over advertisers through entropy-regularized allocation. The resulting
welfare loss is at most \(\beta\log|\mathcal N|\) relative to the optimal
KL-regularized benchmark, and the mechanism approaches first-best welfare as
\(\beta\to0^+\).

\begin{theorem}[Optimization objective and asymptotic efficiency]
\label{thm:latent-advertiser-mixture-efficiency}
Fix a query \(q\) and a type profile \(r\). Let
\((x^\beta,\lambda^\beta)\) be the truthful policy and terminal allocation
induced by LAMA. Then \((x^\beta,\lambda^\beta)\)
maximizes the entropy-relaxed welfare:
\[
\begin{aligned}
    \mathbb E_{\tau\sim x(\cdot\mid q)}
    \Bigg[
        \sum_{i\in\mathcal N}
        \lambda_i(\tau)r_i(\ell(\tau))
        +
        \beta \mathcal H(\lambda(\tau))
        -
        \beta
        \sum_{t=0}^{T_\tau-1}
        \log
        \frac{x(a_t\mid s_t)}
        {\pi_{\mathrm{ref}}(a_t\mid s_t)}
    \Bigg].
\end{aligned}
\]
Moreover, relative to the optimal KL-regularized welfare in the simplex
allocation setting,
\(
    \sup_{\pi,\lambda}
    \mathrm{SW}_\beta(\pi,\lambda\mid q,r),
\)
the welfare loss of the Latent Advertiser Mixture mechanism is bounded by
\[
    0
    \le
    \sup_{\pi,\lambda}
    \mathrm{SW}_\beta(\pi,\lambda\mid q,r)
    -
    \mathrm{SW}_\beta(x^\beta,\lambda^\beta\mid q,r)
    \le
    \beta\log|\mathcal N|.
\]
Thus,
\(
    \lim_{\beta\to 0^+}
    \mathrm{SW}_\beta(x^\beta,\lambda^\beta\mid q,r)
    =
    \max_{\ell\in L(q)}
    \max_{i\in\mathcal N}
    r_i(\ell).
\)
The same welfare-gap bound holds after taking
expectation over \(q\sim\mathcal D\).
\end{theorem}


%% file: sections/implementation.tex
\begin{figure*}[t]
    \centering
    \includegraphics[width=\textwidth]{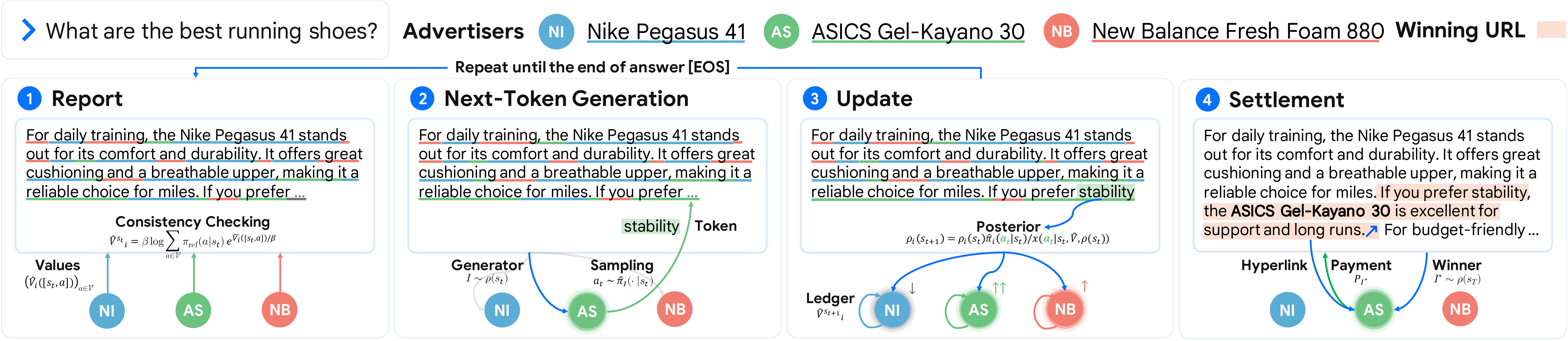}
    \caption{Overview of the Latent Advertiser Mixture Auction (LAMA). (1) \textbf{Report:} Advertisers report soft values for the current child states, and the platform verifies Bellman consistency against the value stored in its ledger. (2) \textbf{Next-token generation:} The platform samples an advertiser from its current belief, whose optimal policy generates the next token. (3) \textbf{Update:} After the context changes, the platform updates the ledger and its belief over the optimal advertiser via Bayes' rule. (4) \textbf{Settlement:} Once the response is complete, the platform samples a winning advertiser from the posterior; the winner pays the corresponding charge and receives the hyperlink placement. See Figure~\ref{fig:lama-presentation-formats} for additional impression formats.}
    \label{fig:lama-overview}
\end{figure*}

In practice, computing the reports required by LAMA is generally beyond the capabilities of
individual advertisers. At each non-terminal state \(s\), advertiser \(i\)
would need to submit the continuation vector
\(
    \left(
        V_i^*([s,a])
    \right)_{a\in\mathcal V}.
\)
Doing so requires both high-quality user feedback data and considerable computational
resources.
We therefore provide reporting as a platform-side service analogous to
autobidding \cite{aggarwal2019autobidding,he2021unified}: the platform learns value reports from advertising feedback and
serves them on advertisers' behalf.

As we will show, the platform can avoid learning the full soft value function through
reinforcement learning by decomposing reporting into two simpler stages. It
learns local soft advantages, which determine changes in value between adjacent
prefixes, and a root value, which anchors the absolute scale. These two outputs
are then assembled into the required report along the realized path. Both stages
can use standard supervised signals and be implemented
by advertiser-conditioned services shared across advertisers.

\subsection{Value Decomposition}
\label{subsec:value-decomposition}
For a prefix
\(s\) and next token \(a\), the optimal soft advantage for advertiser \(i\) satisfies
\(
    A_i^*(s,a)
    =
    V_i^*([s,a])-V_i^*(s)
\)
\footnote{Conventionally, the optimal soft advantage function measures the
incremental value of taking action \(a\) at state \(s\) and then following an
optimal soft policy, relative to the optimal soft value at \(s\). Formally,
\(
    A_i^*(s,a)
    :=
    r_i(s,a)
    +
    \max_{\pi}
    \left\{
        \mathbb E_{\tau\sim\pi(\cdot\mid s')}
        \left[
            \sum_{(u,b)\in\tau} r_i(u,b)
        \right]
        -
        \beta
        \operatorname{\mathbb{D}_{KL}}\!\left(
            \pi(\cdot\mid s')
            \,\middle\|\,
            \pi_{\mathrm{ref}}(\cdot\mid s')
        \right)
    \right\}
    -V_i^*(s)
\)
. In our deterministic token context setting with only terminal rewards, the
maximized continuation objective is \(V_i^*([s,a])\), yielding the displayed
identity.}
.
For a completed response \(y=(a_0,\ldots,a_{T_y-1})\), let
\(
    s_t=[q,a_{<t}],
    \,
    \ell_y=[q,y].
\)
Then the value of every prefix can be obtained from the root value by
telescoping:
\(
    V_i^*(s_t)
    =
    V_i^*(q)
    +
    \sum_{k=0}^{t-1}
    A_i^*(s_k,a_k).
\)
Thus, to construct the truthful child-value report at any reachable state, it is
enough to know the root value \(V_i^*(q)\) and the local advantages
\(A_i^*(s,a)\).

We now present the learning objectives for these two functions.
Define \(r_i(\ell_y)\) as advertiser
\(i\)'s population value for receiving the advertising impression opportunity from response \(y\). 
Let \(\mathcal D\) denote the population comparison distribution over
\((i,q,y,y')\), and use \(\mathcal D(\cdot\mid i,q)\) for its single-response
marginal when only \(y\) is needed.

\paragraph{I. Local advantages.}
For any policy \(\pi_i\), define the cumulative
response score as 
\(
    G_i^\pi(q,y)
    :=
    \beta\log
    \frac{\pi_i(y\mid q)}
         {\pi_{\mathrm{ref}}(y\mid q)}
    =
    \beta\sum_{t=0}^{T_y-1}
    \log
    \frac{\pi_i(a_t\mid s_t)}
         {\pi_{\mathrm{ref}}(a_t\mid s_t)}.
\)
To
learn these token-level log-ratios, the platform samples paired responses
\(y\) and \(y'\) under the same advertiser-query pair \((i,q)\). From their
scalar rewards, it constructs the Bradley-Terry \cite{bradley1952rank} win rate
\[
    \omega_i(q;y,y')
    :=
    \Pr(y\succ y'\mid i,q)
    =
    \sigma\!\left(
        r_i(\ell_y)-r_i(\ell_{y'})
    \right).
\]
Let
\(
    \Delta_i^\pi(q;y,y')
    :=
    G_i^\pi(q,y)-G_i^\pi(q,y').
\)
The ideal local-advantage objective is therefore the binary cross-entropy,
\[
\begin{aligned}
    \mathcal L(\pi)
    :=
    -\mathbb E_{(i,q,y,y')\sim\mathcal D}
    \Big[
        &
        \omega_i(q;y,y')
        \log
        \sigma\!\left(\Delta_i^\pi(q;y,y')\right) \\
        &+
        \bigl(1-\omega_i(q;y,y')\bigr)
        \log
        \sigma\!\left(-\Delta_i^\pi(q;y,y')\right)
    \Big].
\end{aligned}
\]
As we will show later, under the usual regularity conditions, the population
minimizer \(\pi^\dagger\) matches the Bradley-Terry log-odds and therefore
identifies reward differences between completed responses under the same
\((i,q)\). In the soft-control view, the corresponding local log-ratio \(\widehat A_i^\dagger(s,a)
    :=
    \beta\log
    \pi_i^\dagger(a\mid s)/
    \pi_{\mathrm{ref}}(a\mid s)\) recovers
the local soft advantage \(A_i^*(s,a)\).

\paragraph{II. Root values.} The root value is anchored separately from scalar rewards. Since a terminal
state has value \(V_i^*(\ell_y)=r_i(\ell_y)\), telescoping gives, for every
completed response,
\[
    V_i^*(q)
    =
    r_i(\ell_y)
    -
    \sum_{t=0}^{T_y-1}
    A_i^*(s_t,a_t).
\]
Averaging this residual over population responses for the same \((i,q)\) recovers
the root value. The two ideal learning tasks are therefore: learn local
advantages from BT-weighted pairwise differences, and thereby anchor the root value with terminal rewards.

The next theorem shows that the population local-advantage optimum identifies
both the local soft advantages and a residual formula for the root value. The
ledger recursion is then enabled to reconstruct the optimal soft value on every reached
prefix, formally justifying the value decomposition above.

\begin{theorem}[Exact recovery]
\label{thm:value-report-consistency}

Suppose $\mathcal D$ has full comparison support\footnote{for every advertiser-query pair
\((i,q)\) in its marginal support and every distinct pair of feasible completed responses $y,y'$, \(\Pr_{\mathcal{D}}(y,y'|i,q) > 0\).}. 
Then the optimal token policy
$\pi^*$ is the unique minimizer of the
local-advantage objective $\mathcal L$.
Relatedly,
the recovered log-ratio
\(
    \widehat A_i^\dagger(s,a)
    :=
    \beta\log
    \pi_i^\dagger(a\mid s)/\pi_{\mathrm{ref}}(a\mid s)
\)
equals
\(
    A_i^*(s,a).
\)
Also, 
the root value
\(
    v_i^\dagger(q)
    :=
    \mathbb E_{y\sim\mathcal D(\cdot\mid i,q)}
    \left[
        r_i(\ell_y)
        -
        \sum_{t=0}^{T_y-1}
        \widehat A_i^\dagger(s_t,a_t)
    \right]
\)
equals
\(
    V_i^*(q).
\)

Finally, initialize
\(
    \widehat V_i(q):=v_i^\dagger(q)
\)
and recursively define
\(
    \widehat V_i([s,a])
    :=
    \widehat V_i(s)+\widehat A_i^\dagger(s,a).
\)
Then
\(
    \widehat V_i(s)=V_i^*(s)
\)
at every reachable prefix. Consequently, the resulting child-value
vector is exactly the truthful soft-value report required by LAMA.
\end{theorem}

\paragraph{Shared report model.}
Motivated by this recovery result, and to reduce the cost of training and
serving separate models, the platform instantiates the two learned components
with a single shared pretrained backbone. All advertisers share the parameters, while information specific to each advertiser \(i\) enters only through a
conditioning prompt \(\kappa_i\), which contains its identity and
campaign information, such as the creative, landing page, and targeting
context. For any generation prefix \(s\), the same model is run on the
concatenated context \([\kappa_i,s]\) and produces next-token logits
\(z_\theta(a\mid[\kappa_i,s])\). The policy for advertiser \(i\) is the softmax
of these logits:
\[
\begin{aligned}
    \pi_{\theta,i}(a\mid s)
    :=
    \operatorname{softmax}_a z_\theta(a\mid[\kappa_i,s])
    =
    \frac{
        \exp z_\theta(a\mid[\kappa_i,s])
    }{
        \sum_{b\in\mathcal V}
        \exp z_\theta(b\mid[\kappa_i,s])
    } .
\end{aligned}
\]
And the root anchor is produced by a fine-tuned scalar value head attached to the same shared
representation:
\[
    v_{\phi,i}(q)
    :=
    v_\phi([\kappa_i,q]).
\]
The deployed object is therefore a single shared model that produces
advertiser-specific outputs: next-token logits for local reports and a scalar
root value for initializing the ledger. Appendix~\ref{app:practical-training}
gives the practical finite-sample training procedure in
Algorithm~\ref{alg:report-model-training}.

\subsection{Serving-time Procedure}

Figure~\ref{fig:lama-overview} provides an intuitive overview of LAMA. At serving time, the learned objects above are exposed as a single report
service. For each query \(q\), the platform first retrieves a small candidate
set of advertisers through the standard advertising funnel. For each retrieved
advertiser \(i\), it builds the advertiser prefix \(\kappa_i\), evaluates the
root value, and initializes the value ledger by
\(
    \widehat V_i(s_0)
    :=
    v_{\phi,i}(q).
\)
The same root values initialize the allocation posterior. At each non-terminal
prefix \(s_t\), the service returns advertiser-conditioned logits. The platform
converts them into local advantages
\(
    \widehat A_{\theta,i}(s_t,a)
    =
    \beta
    \log(
        \pi_{\theta,i}(a\mid s_t)
        /
        \pi_{\mathrm{ref}}(a\mid s_t)
    )
\)
and therefore into child-value reports
\(
    \widehat V_i^{s_t}([s_t,a])
    =
    \widehat V_i(s_t)+\widehat A_{\theta,i}(s_t,a)
\).
\footnote{This report is Bellman-consistent, regardless of whether the learned model is already optimal. To see this, since
$
    \sum_{a\in\mathcal V}
    \pi_{\mathrm{ref}}(a\mid s_t)
    \exp(
        \widehat A_{\theta,i}(s_t,a)/\beta
    )
    =
    \sum_{a\in\mathcal V}
    \pi_{\theta,i}(a\mid s_t)
    =
    1,
$
we therefore have
$
    \beta
    \log
    \sum_{a\in\mathcal V}
    \pi_{\mathrm{ref}}(a\mid s_t)
    \exp(
        \widehat V_i^{s_t}([s_t,a])/\beta
    )
     =
    \beta
    \log
    [
        \exp(
            \widehat V_i(s_t)/\beta
        )
        \sum_{a\in\mathcal V}
        \pi_{\mathrm{ref}}(a\mid s_t)
        \exp(
            \widehat A_{\theta,i}(s_t,a)/\beta
        )
    ]  =
    \widehat V_i(s_t).
$
}
The remaining steps follow the theoretical mechanism directly. Since
the candidate advertiser set is typically small, LAMA incurs only a small
constant-factor overhead in language model forward passes relative to generating a reference
answer, keeping online serving lightweight.
Algorithm~\ref{alg:lama}
formally summarizes the serving procedure.

\begin{algorithm}[t]
\caption{Latent Advertiser Mixture Auction (LAMA)}
\label{alg:lama}
\begin{algorithmic}[1]
\Require Query $q$, Candidate advertisers $\mathcal{N} = \{1, \dots, n\}$, Reference policy $\pi_{\text{ref}}$, Shared report model $(\pi_\theta, v_\phi)$, Temperature $\beta$
\State \textbf{Initialize:} Context $s_0 \leftarrow q$, Step $t \leftarrow 0$, Allocation posterior $\rho_i(s_0) \leftarrow \frac{\exp(v_{\phi, i}(q) / \beta)}{\sum_{j \in \mathcal{N}} \exp(v_{\phi, j}(q) / \beta)}$ for all $i \in \mathcal{N}$
\State \textbf{Initialize Ledger:} $\hat{V}_i(s_0) \leftarrow v_{\phi, i}(q)$ for all $i \in \mathcal{N}$

\While{$s_t$ is not terminal}
    \For{each advertiser $i \in \mathcal{N}$}
        \State \textcolor{blue}{\textsc{// Compute local advantage and consistent child-value reports}}
        \State $\hat{A}_{\theta, i}(s_t, a) \leftarrow \beta \log \frac{\pi_{\theta, i}(a \mid s_t)}{\pi_{\text{ref}}(a \mid s_t)} \quad \forall a \in \mathcal{V}$
        \State $\hat{V}_i^{s_t}([s_t, a]) \leftarrow \hat{V}_i(s_t) + \hat{A}_{\theta, i}(s_t, a) \quad \forall a \in \mathcal{V}$
    \EndFor

    \State \textcolor{blue}{\textsc{// Hierarchical Sampling Process}}
    \State Sample a latent advertiser $I \sim \rho(s_t)$ \Comment{Choose guiding advertiser} 
    \State Sample next token $a_t \sim \hat{\pi}_I(\cdot \mid s_t)$ \Comment{Generate token from that advertiser's policy}
    
    \State \textcolor{blue}{\textsc{// Update state and Bayesian allocation posterior}}
    \State $s_{t+1} \leftarrow [s_t, a_t]$
    \For{each advertiser $i \in \mathcal{N}$}
        \State $\rho_i(s_{t+1}) \leftarrow \frac{\rho_i(s_t) \cdot \pi_{\theta, i}(a_t \mid s_t)}{x(a_t \mid s_t, \hat{V}^{s_t}, \rho(s_t))}$ \Comment {Posterior update}
        \State $\hat{V}_i(s_{t+1}) \leftarrow \hat{V}_i(s_t) + \hat{A}_{\theta, i}(s_t, a_t)$ \Comment{Update value ledger}
    \EndFor
    \State $t \leftarrow t + 1$
\EndWhile

\State \textbf{Terminal State Reached:} Let $s_T$ be the terminal state, Total tokens generated $T$
\State \textcolor{blue}{\textsc{// Winner-Pay Settlement: Sample a single winner from the final posterior}}
\State Sample a winner $I^* \sim \rho(s_T)$

\For{each advertiser $i \in \mathcal{N}$}
    \If{$i = I^*$}
        \State \textcolor{blue}{\textsc{// Telescoped trajectory-level payment}}
        \State $P_i \leftarrow \rho_i(s_T)\hat{V}_i(s_T) - \beta \log \sum_{j \in \mathcal{N}} \exp(\hat{V}_j(s_0)/\beta) + \beta \log \left( 1 + \sum_{j \neq i} \exp(\widehat{V}_j(s_0)/\beta) \right)$
        \State $P_i \leftarrow {\rho_i(s_T)}^{-1} P_i$ \Comment{Adjust by winning probability}
    \Else
        \State $P_i \leftarrow 0$ \Comment{Non-winners pay nothing}
    \EndIf
\EndFor

\State \Return Generated sequence $s_T$, Winner $I^*$, Payments $(P_1, \dots, P_n)$
\end{algorithmic}
\end{algorithm}


Notably, the fractional form of LAMA may
charge an advertiser that does not ultimately receive the advertising
opportunity. We therefore use the equivalent randomized settlement in
Algorithm~\ref{alg:lama}: after the final posterior is formed, the platform
samples the winner \(I^*\sim\rho(s_T)\), settles only with that winner, and
charges zero to all non-winners. Because this preserves expected
utilities, the incentive-compatibility guarantees carry over directly. It also
yields outcome IR while preserving expected weak budget balance, even though
individual payment transfers may have either sign.


\begin{proposition}[Outcome IR and Expected WBB]
\label{prop:lama-winner-payment-ir-wbb}
If the served reports are exact, then the
winner-settled LAMA implementation in Algorithm~\ref{alg:lama} is outcome
individually rational: for every realized terminal state and sampled winner,
non-winners pay zero and the winning advertiser \(i\) obtains non-negative realized utility. 
Moreover, conditional on the query and candidate advertiser set, the platform is
weakly budget balanced in expectation:
\(
    \mathbb E
    [
        \sum_{i\in\mathcal N}
        P_i(s_T,I^*)
    ]
    \ge 0.
\)
\end{proposition}


%% file: sections/exp.tex
In this section, we empirically evaluate the proposed LAMA in generation-native advertising settings. The goal is to test whether token-level participation of advertisers can improve monetization while preserving the quality and naturalness of generated responses. 

\subsection{Setup}
We evaluate mechanisms across representative advertising verticals. Each vertical consists of a set of user queries, candidate advertisers and advertiser descriptions. We use a reference language model (Qwen3-14B \cite{yang2025qwen3}), advertiser-shared report models (trained LoRA heads \cite{hu2021lora}), and a prediction model for oracle impression value \cite{lv12_esci_msmarco_minilm_l12_v2,qwen3embedding}. The main experiments cover three verticals constructed from real-world commercial search query splits in the Webis Generated Native Ads 2024 dataset \cite{schmidt2024detecting}: \textbf{Workout}, \textbf{Vacation}, and \textbf{Car}. Each vertical contains three heterogeneous advertisers.

\paragraph{Baselines.} We first compare against six heuristic incentive compatible baselines obtained by combining two auction rules with three generation rules.
The auction rules are: 
\textbf{(i)} \textsc{Before}. The winner is chosen before generation uniformly at random, and the advertiser pays nothing.
\textbf{(ii)} \textsc{After}. The candidate responses are generated first. Then advertisers are scored using predicted impression value, and a second-price auction is applied.
The generation rules are:
\textbf{(i)} \textsc{Original}. Generate only with the reference model.
\textbf{(ii)} \textsc{Edit}. Use a no-ad reference response as the base and insert sponsored content of an advertiser through a fused editing process.
\textbf{(iii)} \textsc{Policy}. Each advertiser submits its preferred response, i.e., the content it expects to maximize its impression value. We approximate this response using best-of-K sampling from the advertiser's optimal policy, with candidates ranked by self-predicted impression value.

We also compare against MOSAIC \cite{DBLP:journals/corr/abs-2405-05905}, a representative response-level aggregation mechanism. It generates $M$ complete context-aware candidate responses, scores each candidate by aggregate advertiser value, and selects a final response via a weighted lottery.

\paragraph{Metrics.} 
Our evaluation considers three perspectives: the platform, the advertiser, and the user. On the platform side, we report \textbf{(i)} \textsc{Welfare}, defined as the realized advertising value minus the KL regularization term that captures the user-experience cost, and \textbf{(ii)} \textsc{Revenue}, defined as the payment collected from the winning advertiser. On the advertiser side, we measure \textbf{(i)} \textsc{Value}, namely the impression value obtained by the winning advertiser. On the user side, we evaluate \textbf{(i)} \textsc{Quality}. Since the KL penalty alone does not fully reflect the actual user experience of the generated text, we adopt the evaluation rubrics from the GEM-Bench \cite{hu2026gembenchbenchmarkadinjectedresponse} to assess user utility. Specifically, \textsc{Quality} aggregates multiple dimensions, including answer relevance, informativeness, fluency and coherence, ad naturalness, ad factuality, and intrusiveness.

\subsection{Main Results}

\definecolor{darkgreen}{RGB}{90,121,62}
        
        
        
        
        
\begin{table}[t!]
    \centering
    \caption{
        Main results on three representative verticals from Webis. Smaller gray values denote 95\% bootstrap errors.
    }
    \label{tab:main_exp}

    \newcommand{\result}[2]{%
        \begin{tabular}[c]{@{}c@{}}
            #1 \\[-2pt]
            {\scriptsize\textcolor{gray}{$\pm\,#2$}}
        \end{tabular}%
    }
    \newcommand{\bestresult}[2]{%
        \begin{tabular}[c]{@{}c@{}}
            \textbf{#1} \\[-2pt]
            {\scriptsize\textcolor{gray}{$\pm\,#2$}}
        \end{tabular}%
    }

    \begin{adjustbox}{max width=\textwidth}
    \begin{tabular}{l c c c c}
        \toprule
        \multirow{2}{*}{\textbf{Method}}
        & \multicolumn{2}{c}{\textsc{Platform}}
        & \textsc{User}
        & \textsc{Advertiser} \\

        \cmidrule(lr){2-3}
        \cmidrule(lr){4-4}
        \cmidrule(lr){5-5}

        & Wel.\,\textcolor{gray}{$\uparrow$}
        & Rev.\,\textcolor{gray}{$\uparrow$}
        & Qua.\,\textcolor{gray}{$\uparrow$}
        & Val.\,\textcolor{gray}{$\uparrow$} \\

        \midrule
        \multicolumn{5}{c}{\textbf{Allocate before generation}} \\
        \midrule

        Reference
        & \result{0.2945}{0.0117}
        & \result{0.0000}{0.0000}
        & \result{52.9429}{0.2753}
        & \result{0.2945}{0.0116} \\

        Edit
        & \result{0.0297}{0.0138}
        & \result{0.0000}{0.0000}
        & \result{65.7685}{0.2323}
        & \result{0.2598}{0.0133} \\

        Policy
        & \result{0.4562}{0.0070}
        & \result{0.0000}{0.0000}
        & \result{65.1635}{0.2230}
        & \result{0.8225}{0.0045} \\

        \midrule
        \multicolumn{5}{c}{\textbf{Allocate after generation}} \\
        \midrule

        Reference
        & \result{0.3558}{0.0134}
        & \result{0.2875}{0.0112}
        & \result{53.8997}{0.2642}
        & \result{0.3558}{0.0132} \\

        Edit
        & \result{0.1772}{0.0135}
        & \result{0.2490}{0.0117}
        & \result{66.4785}{0.2423}
        & \result{0.4046}{0.0132} \\

        Policy
        & \result{0.5080}{0.0041}
        & \result{0.7501}{0.0053}
        & \result{65.5939}{0.2197}
        & \result{0.8253}{0.0044} \\

        MOSAIC
        & \result{0.4390}{0.0118}
        & \result{0.5274}{0.0121}
        & \result{60.4100}{0.3236}
        & \result{0.5550}{0.0118} \\

        \midrule
        \rowcolor{green!15}
        \textbf{LAMA (Ours)}
        & \bestresult{0.5205}{0.0053}
        & \bestresult{0.8305}{0.0049}
        & \bestresult{66.5239}{0.2278}
        & \bestresult{0.8568}{0.0043} \\

        \bottomrule
    \end{tabular}
    \end{adjustbox}
\end{table}

Table~\ref{tab:main_exp} summarizes the main results. LAMA achieves the strongest overall performance among the compared methods, attaining the highest mean platform welfare ($0.5205$), revenue ($0.8305$), advertiser value ($0.8568$), and user quality ($66.5239$). Relative to the allocate-after policy baseline, LAMA delivers a clear improvement in revenue, together with further gains in welfare and advertiser value. At the same time, it preserves user-facing response quality, matching the best-performing baseline on this dimension. 
Appendix~\ref{app:vertical-advertiser-overview} provides additional experimental evidence, including case studies and token-level analyses for complementary and competing advertisers, latency measurements, value prediction accuracy and calibration, and bid-offset sweeps of allocation, payment, and advertiser utility.

Taken together, these results provide an initial proof of concept for LAMA: token-level advertiser participation can improve monetization without sacrificing response quality, suggesting the potential of generation-native advertising as a viable paradigm.

%% file: sections/related.tex
Classical online advertising auctions typically allocate exogenously defined inventory, such as sponsored-search or display slots, through mechanisms including GSP and VCG \cite{Edelman_2007,Varian_2007}. Subsequent work incorporates cascade effects, ad externalities, and richer click-through-rate structures \cite{DBLP:conf/wine/AggarwalFMP08,DBLP:conf/ecai/0001RSV16,DBLP:conf/wine/CavalloW14}, and studies the joint optimization of advertising and organic content across search, recommendation, e-commerce, and contextual advertising \cite{Carrion_2023,Li_2023,An_2025,DBLP:conf/www/GhoshM08,DBLP:conf/wine/GiotisK08,DBLP:journals/tcs/0001RSV18}. Despite richer interactions with users and surrounding content, these models largely retain a common abstraction: the mechanism allocates pre-existing slots, items, or positions whose advertising opportunities are defined before allocation.

Generative interfaces challenge this abstraction because advertising opportunities may depend on the response being generated. More broadly, recent work studies mechanisms in which strategic agents influence LLM outputs through next-token distributions or preferences, making generation itself part of the mechanism outcome \cite{DBLP:conf/www/DuttingMLXZ24,DBLP:conf/kdd/Dubey0KM024,DBLP:journals/corr/abs-2405-05905}. In LLM advertising, some works define opportunities over discourse segments or context-dependent position--creative pairs \cite{NEURIPS2024_20dcab0f,balseiro2025positionauctionsaigeneratedcontent}, while others study response-level ad insertion \cite{xu2026adinsertionllmgeneratedresponses}. More LLM-native approaches incorporate response-quality constraints \cite{han2026mechanismdesignqualitypreservingllmadvertising}, allocate advertising through the model's output distribution \cite{zhao2026llmauctiongenerativeauctionllmnative}, or explore opportunities arising from retrieval decisions and internal model representations \cite{sun2026lerallmenhancedragad,yun2026llmadvertisementbasedneuron}. Together, these works progressively move advertising allocation from fixed slots toward the generation process itself.

Our work takes a further step toward generation-native advertising by allowing advertisers to influence generation at token-level granularity. Rather than allocating an opportunity identified before or after generation, LAMA operates along the generation trajectory: advertiser influence evolves with the current prefix, posterior beliefs, and continuation values. Thus, the advertising opportunity itself emerges dynamically through generation, shifting the mechanism-design problem from allocating sponsored content within generated responses to governing how advertiser influence shapes the response-generation process.

%% file: sections/training_algorithm.tex
Algorithm~\ref{alg:report-model-training} gives a practical implementation
of the decomposed two learning tasks in Section~\ref{subsec:value-decomposition}.

\begin{algorithm}[t]
\caption{Practical Training of the Shared Report Model}
\label{alg:report-model-training}
\begin{algorithmic}[1]
\Require Advertiser-Query pairs \(\mathcal D_{IQ}\), Advertiser prefix \(\{\kappa_i\}_{i\in\mathcal N}\), Reference policy \(\pi_{\mathrm{ref}}\), Rollout budget \(B\), Temperature \(\beta\)

\State \textcolor{blue}{\textsc{// Construct pairwise comparison data}}
\State \(\widehat{\mathcal D}_{\mathrm{pair}}\leftarrow\varnothing,\,\widehat{\mathcal D}_{\mathrm{root}}\leftarrow\varnothing\) \Comment{Initialize training datasets}
\For{each advertiser-query pair \((i,q)\in\mathcal D_{IQ}\)}
    \For{\(b=1,\dots,B\)}
        \State \(y_{i,q}^{(b)}\sim\pi_{\mathrm{ref}}(\cdot\mid[\kappa_i,q])\) \Comment{Sample responses}
        \State \(r_i(\ell_{y_{i,q}^{(b)}})\leftarrow\operatorname{OnlineTest}(i,q,y_{i,q}^{(b)})\) \Comment{Estimate population impression value from online experiments}
    \EndFor
    \For{every \(1\le b<b'\le B\)} \Comment{Full comparison}
        \State \(\omega_i^{(b,b')}(q)\leftarrow\sigma\!\left(r_i(\ell_{y_{i,q}^{(b)}})-r_i(\ell_{y_{i,q}^{(b')}})\right)\)
        \State \(\widehat{\mathcal D}_{\mathrm{pair}}\leftarrow\widehat{\mathcal D}_{\mathrm{pair}}\cup\left\{\bigl(i,q,y_{i,q}^{(b)},y_{i,q}^{(b')},\omega_i^{(b,b')}(q)\bigr)\right\}\)
    \EndFor
\EndFor

\State \textcolor{blue}{\textsc{// Train local advantages}}
\State \(\pi_\theta\leftarrow\pi_{\mathrm{ref}}\) \Comment{Initialize advertiser-shared policy model}
\State \(\displaystyle
    \Delta_i^{\pi_\theta}(q;y,y')
    \leftarrow
    \beta \left[
        \begin{aligned}
        \sum_{t=0}^{T_y-1}
        &\log\frac{\pi_{\theta,i}(a_t\mid [\kappa_i,s_t])}{\pi_{\mathrm{ref}}(a_t\mid [\kappa_i,s_t])}
        \\&-\sum_{t=0}^{T_{y'}-1}
        \log\frac{\pi_{\theta,i}(a'_t\mid [\kappa_i,s'_t])}{\pi_{\mathrm{ref}}(a'_t\mid [\kappa_i,s'_t])}
        \end{aligned}
    \right]
    \)
\State \(\displaystyle
    \widehat{\mathcal L}(\pi_\theta)
    \leftarrow
    \frac{1}{|\widehat{\mathcal D}_{\mathrm{pair}}|}
    \sum_{\substack{(i,q,y,y',\omega)\\\in\widehat{\mathcal D}_{\mathrm{pair}}}}
    \left [\begin{aligned}
        &\omega\log\sigma\!\left(\Delta_i^{\pi_\theta}(q;y,y')\right)\\[-0.2em]
        &+(1-\omega)\log\sigma\!\left(-\Delta_i^{\pi_\theta}(q;y,y')\right)
    \end{aligned} \right]
    \)
\State \(\theta^\dagger\leftarrow\arg\min_\theta\widehat{\mathcal L}(\pi_\theta)\) \Comment{Optimize local advantages}

\State \textcolor{blue}{\textsc{// Learn root values from residual targets}}
\For{each advertiser-query pair \((i,q)\in\mathcal D_{IQ}\)}
    \For{\(b=1,\dots,B\)}
        \State \(\widetilde y_{i,q}^{(b)}\sim\pi_{\theta^\dagger,i}(\cdot\mid [\kappa_i, q])\) \Comment{Sample on-policy responses}
        \State \(r_i(\ell_{\widetilde y_{i,q}^{(b)}})\leftarrow\operatorname{OnlineTest}(i,q,\widetilde y_{i,q}^{(b)})\)
        \State \(\displaystyle \widetilde v_{i,q}^{(b)}\leftarrow r_i(\ell_{\widetilde y_{i,q}^{(b)}})-\beta\sum_{t=0}^{T_{\widetilde y}-1}\log\frac{\pi_{\theta^\dagger,i}(a_t\mid [\kappa_i,s_t])}{\pi_{\mathrm{ref}}(a_t\mid [\kappa_i,s_t])}\)
        \State \(\widehat{\mathcal D}_{\mathrm{root}}\leftarrow\widehat{\mathcal D}_{\mathrm{root}}\cup\left\{\bigl(i,q,\widetilde v_{i,q}^{(b)}\bigr)\right\}\)
    \EndFor
\EndFor
\State \(v_\phi\leftarrow\operatorname{AttachValueHead}(\pi_{\theta^\dagger})\) \Comment{Initialize root-value model}
\State \(\displaystyle \phi^\dagger\leftarrow\arg\min_\phi\frac{1}{|\widehat{\mathcal D}_{\mathrm{root}}|}\sum_{(i,q,\widetilde v)\in\widehat{\mathcal D}_{\mathrm{root}}}\bigl(v_{\phi,i}(q)-\widetilde v\bigr)^2\)

\State \Return Shared report model \(\bigl(\pi_{\theta^\dagger},v_{\phi^\dagger}\bigr)\)
\end{algorithmic}
\end{algorithm}

%% file: sections/potential_business_format.tex
Figure~\ref{fig:lama-presentation-formats} presents three representative business formats for applying LAMA in the single-winner setting: a hyperlink embedded in the generated response, a clickable response block, and a sponsored card displayed alongside the response.

\begin{figure*}[t]
    \centering
    \begin{minipage}[t]{0.31\textwidth}
        \centering
        \includegraphics[width=\linewidth]{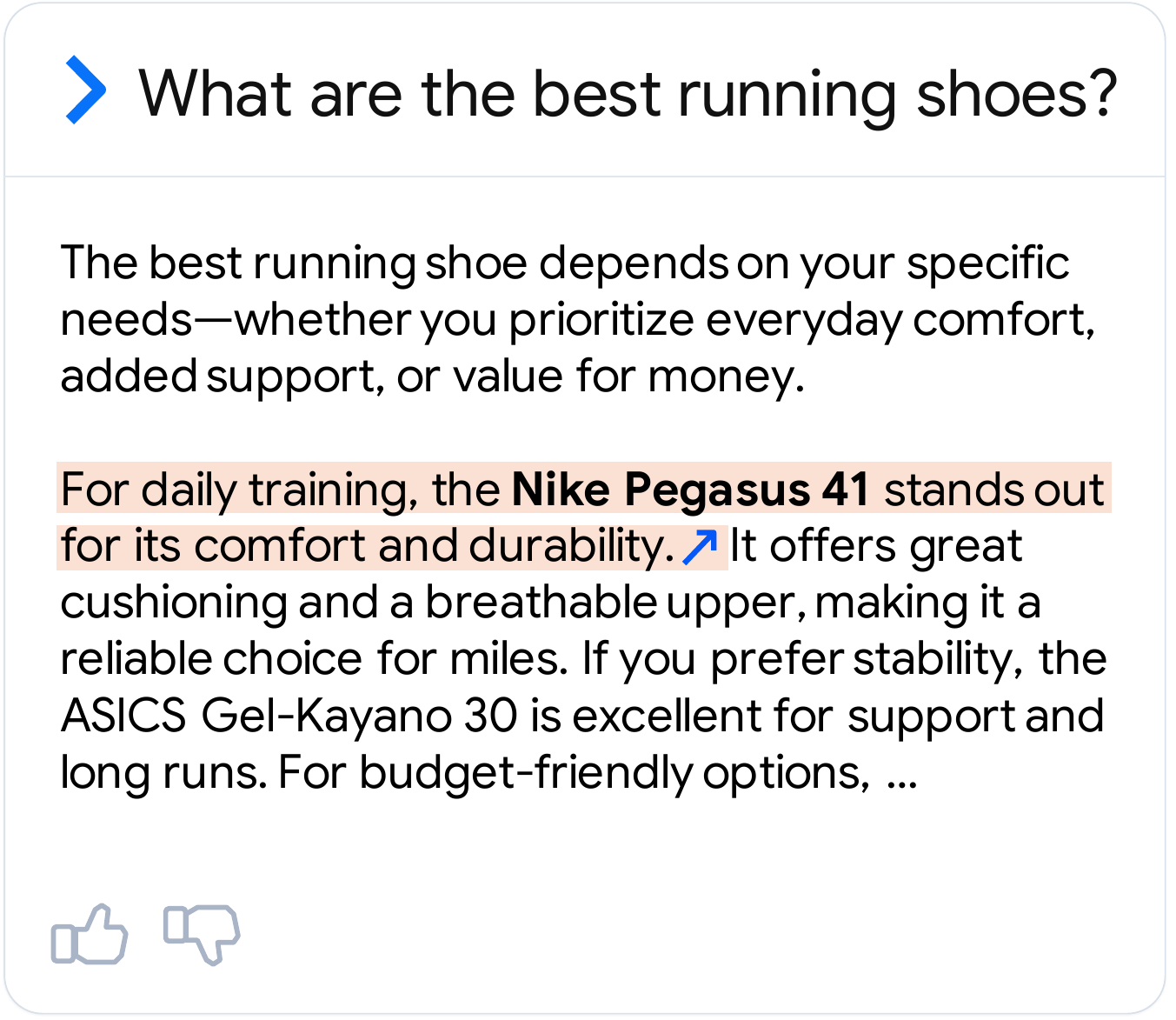}
        \par\smallskip
        \textbf{(a)} Hyperlink
    \end{minipage}%
    \hspace{0.02\textwidth}%
    \begin{minipage}[t]{0.31\textwidth}
        \centering
        \includegraphics[width=\linewidth]{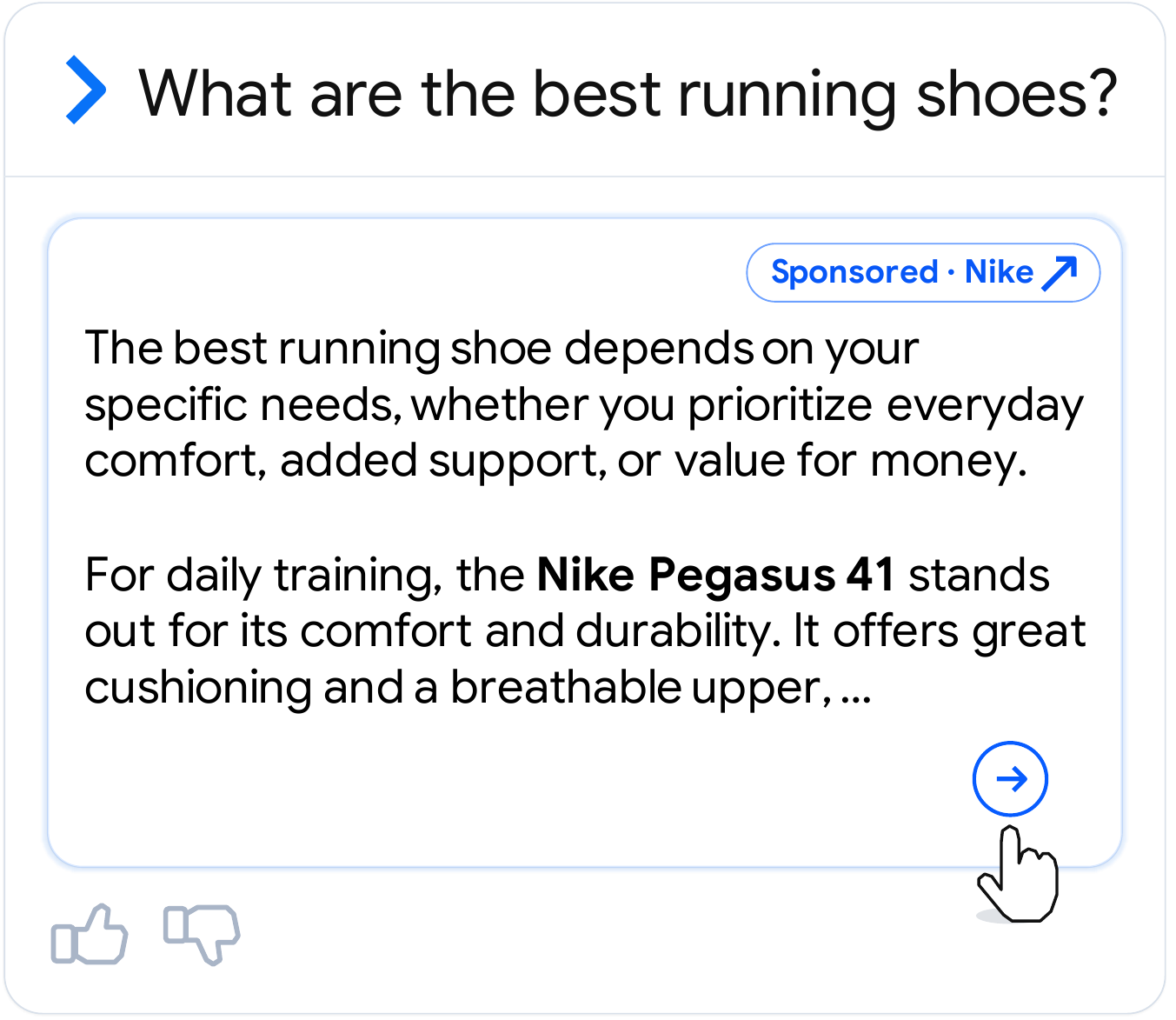}
        \par\smallskip
        \textbf{(b)} Clickable Block
    \end{minipage}%
    \hspace{0.02\textwidth}%
    \begin{minipage}[t]{0.31\textwidth}
        \centering
        \includegraphics[width=\linewidth]{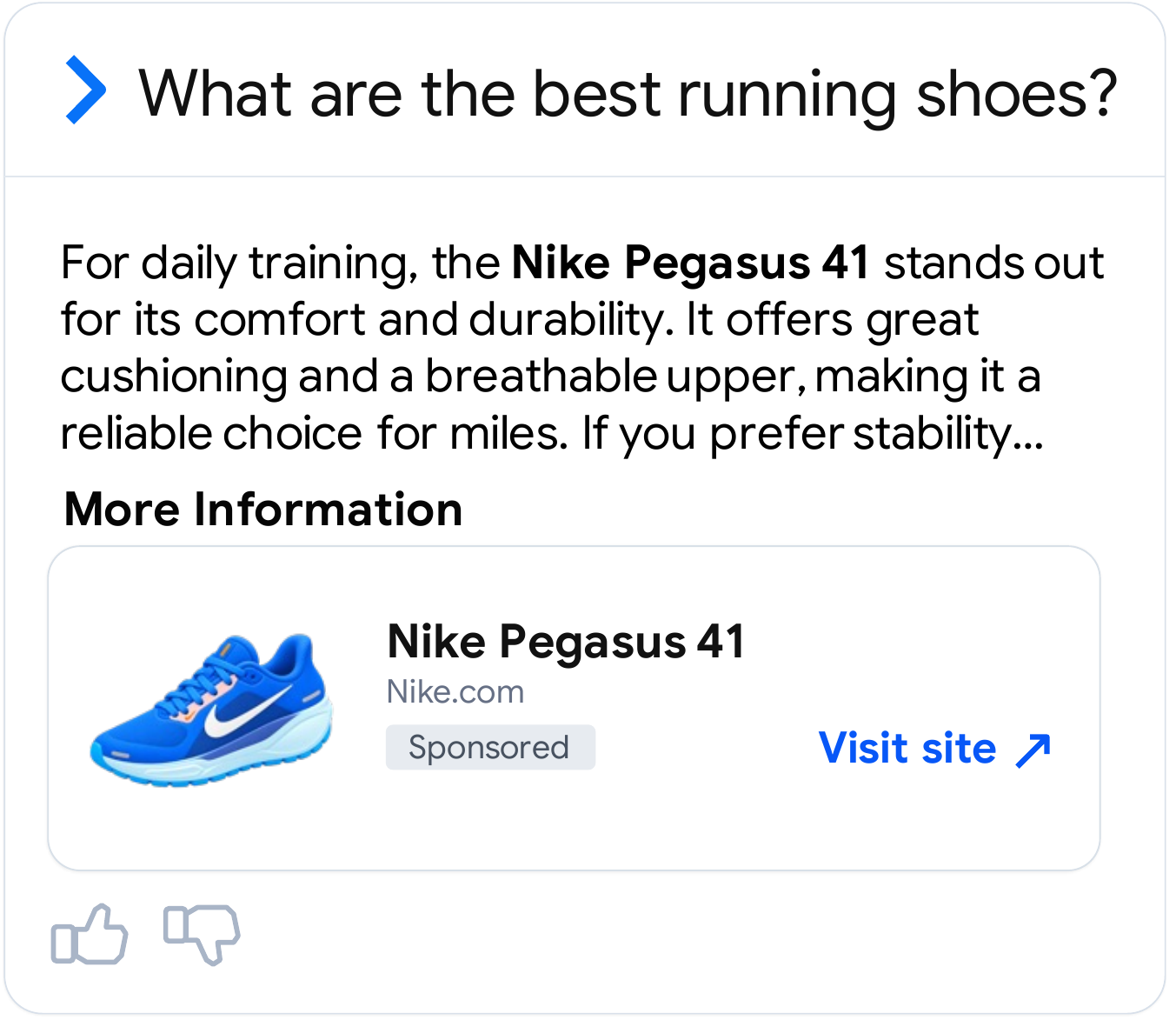}
        \par\smallskip
        \textbf{(c)} Sponsored Card
    \end{minipage}
    \caption{Potential commercialization formats for LAMA. \textbf{(a)} Hyperlink: the winning advertiser's brand name in the generated response is hyperlinked to its landing page; this is the format used in our experiments. \textbf{(b)} Clickable Block: the entire response block links to the advertiser's landing page. \textbf{(c)} Sponsored Card: an advertisement card is displayed below the generated response, a format currently adopted by platforms such as ChatGPT Go, Google AI Overviews, and Baidu AI Smart Answers.}
    \label{fig:lama-presentation-formats}
\end{figure*}

%% file: sections/proof.tex
\begin{lemma}[Optimal soft Bellman recursion]
\label{lem:optimal-soft-bellman}
Consider a finite-horizon token MDP with deterministic transition
\(s_{t+1}=[s_t,a_t]\), reference policy \(\pi_{\mathrm{ref}}\), and terminal
reward \(r:\mathcal S_T\to\mathbb R\). For every state \(s\), define the optimal
KL-regularized continuation value
\[
    V^*(s)
    :=
    \sup_{\pi}
    \mathbb E_{\tau\sim\pi(\cdot\mid s)}
    \left[
        r(s_T)
        -
        \beta
        \sum_{t=0}^{T_\tau-1}
        \log
        \frac{\pi(a_t\mid s_t)}
        {\pi_{\mathrm{ref}}(a_t\mid s_t)}
    \right].
\]
Assume \(\pi_{\mathrm{ref}}(a\mid s)>0\) for every feasible action \(a\) at
every non-terminal state \(s\). Then \(V^*\) satisfies the terminal boundary
condition
\(
    V^*(\ell)=r(\ell),
    \,\ell\in\mathcal S_T,
\)
and the soft Bellman recursion
\[
    V^*(s)
    =
    \beta
    \log
    \sum_{a\in\mathcal V}
    \pi_{\mathrm{ref}}(a\mid s)
    \exp\left(
        \frac{V^*([s,a])}{\beta}
    \right),
    \, s\notin\mathcal S_T.
\]
Moreover, the optimal continuation policy is
\[
    \pi^*(a\mid s)
    =
    \frac{
        \pi_{\mathrm{ref}}(a\mid s)
        \exp\left(
            V^*([s,a])/\beta
        \right)
    }{
        \sum_{a'\in\mathcal V}
        \pi_{\mathrm{ref}}(a'\mid s)
        \exp\left(
            V^*([s,a'])/\beta
        \right)
    }.
\]
\end{lemma}
\begin{proof}
The terminal condition is immediate. If \(s=\ell\in\mathcal S_T\), the process has
already terminated and no further KL penalty is incurred, so
\[
    V^*(\ell)=r(\ell).
\]

Now fix a non-terminal state \(s\). Any continuation policy can be decomposed
into its first-step action distribution
\[
    \mu(a):=\pi(a\mid s)\in\Delta(\mathcal V)
\]
and continuation policies after each child state \([s,a]\). Hence, by the
principle of optimality,
\[
    V^*(s)
    =
    \max_{\mu\in\Delta(\mathcal V)}
    \left\{
        \sum_{a\in\mathcal V}
        \mu(a)V^*([s,a])
        -
        \beta
        \sum_{a\in\mathcal V}
        \mu(a)
        \log
        \frac{\mu(a)}
        {\pi_{\mathrm{ref}}(a\mid s)}
    \right\}.
\]
This is a strictly concave optimization problem over the simplex. Introduce a
Lagrange multiplier \(\xi\) for the constraint
\(\sum_a\mu(a)=1\). The Lagrangian is
\[
\begin{aligned}
\mathcal L(\mu,\xi)
= &
\sum_{a\in\mathcal V}
\mu(a)V^*([s,a])
-
\beta
\sum_{a\in\mathcal V}
\mu(a)
\log
\frac{\mu(a)}
{\pi_{\mathrm{ref}}(a\mid s)} \\
&+
\xi
\left(
\sum_{a\in\mathcal V}\mu(a)-1
\right).
\end{aligned}
\]
For every action \(a\) with \(\mu(a)>0\), the first-order condition gives
\[
    \frac{\partial \mathcal L}{\partial \mu(a)}
    =
    V^*([s,a])
    -
    \beta
    \left(
        \log
        \frac{\mu(a)}
        {\pi_{\mathrm{ref}}(a\mid s)}
        +1
    \right)
    +
    \xi
    =
    0.
\]
Therefore
\[
    \log
    \frac{\mu(a)}
    {\pi_{\mathrm{ref}}(a\mid s)}
    =
    \frac{V^*([s,a])+\xi-\beta}{\beta},
\]
or equivalently
\[
    \mu(a)
    =
    \pi_{\mathrm{ref}}(a\mid s)
    \exp\left(
        \frac{V^*([s,a])}{\beta}
    \right)
    \exp\left(
        \frac{\xi-\beta}{\beta}
    \right).
\]
The final exponential factor is independent of \(a\), so normalization gives
\[
    \mu^*(a)
    =
    \frac{
        \pi_{\mathrm{ref}}(a\mid s)
        \exp\left(
            V^*([s,a])/\beta
        \right)
    }{
        \sum_{a'\in\mathcal V}
        \pi_{\mathrm{ref}}(a'\mid s)
        \exp\left(
            V^*([s,a'])/\beta
        \right)
    }.
\]
This is the optimal continuation policy \(\pi^*(a\mid s)\).

It remains to compute the optimal value. Let
\[
    Z(s)
    :=
    \sum_{a\in\mathcal V}
    \pi_{\mathrm{ref}}(a\mid s)
    \exp\left(
        \frac{V^*([s,a])}{\beta}
    \right).
\]
From the expression for \(\mu^*\),
\[
    \log
    \frac{\mu^*(a)}
    {\pi_{\mathrm{ref}}(a\mid s)}
    =
    \frac{V^*([s,a])}{\beta}
    -
    \log Z(s).
\]
Substituting this identity into the one-step objective gives
\[
\begin{aligned}
    V^*(s)
    &=
    \sum_a \mu^*(a)V^*([s,a])
    -
    \beta
    \sum_a \mu^*(a)
    \log
    \frac{\mu^*(a)}
    {\pi_{\mathrm{ref}}(a\mid s)}
    \\
    &=
    \sum_a \mu^*(a)V^*([s,a])
    -
    \beta
    \sum_a \mu^*(a)
    \left(
        \frac{V^*([s,a])}{\beta}
        -
        \log Z(s)
    \right)
    \\
    &=
    \beta\log Z(s).
\end{aligned}
\]
Hence
\[
    V^*(s)
    =
    \beta
    \log
    \sum_{a\in\mathcal V}
    \pi_{\mathrm{ref}}(a\mid s)
    \exp\left(
        \frac{V^*([s,a])}{\beta}
    \right),
\]
which proves the soft Bellman recursion.
\end{proof}

\begin{lemma}[Convexity and gradient of weighted log-sum-exp]\label{lem:logsumexp-convex}
Let \(\beta>0\), let \(\mathcal K\) be a finite index set, and let
\(\{w_k\}_{k\in\mathcal K}\) be nonnegative weights, not all zero. Define
\[
    F_\beta(x)
    =
    \beta \log
    \left(
        \sum_{k\in\mathcal K}
        w_k \exp\left(\frac{x_k}{\beta}\right)
    \right),
    \qquad x\in \mathbb R^{\mathcal K}.
\]
Then \(F_\beta\) is convex on \(\mathbb R^{\mathcal K}\). Moreover, \(F_\beta\)
is differentiable and, for each \(k\in\mathcal K\),
\[
    \frac{\partial F_\beta(x)}{\partial x_k}
    =
    \frac{
        w_k \exp(x_k/\beta)
    }{
        \sum_{k'\in\mathcal K}
        w_{k'} \exp(x_{k'}/\beta)
    }.
\]
\end{lemma}

\begin{proof}
Let
\[
    Z_\beta(x)
    =
    \sum_{k\in\mathcal K}
    w_k \exp\left(\frac{x_k}{\beta}\right),
\]
and define
\[
    p_k(x)
    =
    \frac{
        w_k \exp(x_k/\beta)
    }{
        Z_\beta(x)
    }.
\]
Since the weights are nonnegative and not all zero, \(Z_\beta(x)>0\) for every
\(x\). Hence \(p(x)=(p_k(x))_{k\in\mathcal K}\) is a probability vector over
\(\mathcal K\). Direct differentiation gives
\[
    \frac{\partial F_\beta(x)}{\partial x_k}
    =
    \beta \cdot
    \frac{1}{Z_\beta(x)}
    \cdot
    w_k \cdot \frac{1}{\beta}
    \exp\left(\frac{x_k}{\beta}\right)
    =
    p_k(x).
\]
Therefore,
\[
    \nabla F_\beta(x)=p(x).
\]

We next compute the Hessian. For any \(k,k'\in\mathcal K\),
\[
    \frac{\partial p_k(x)}{\partial x_{k'}}
    =
    \frac{1}{\beta}
    p_k(x)
    \left(
        \mathbf 1\{k=k'\}-p_{k'}(x)
    \right).
\]
Thus
\[
    \nabla^2 F_\beta(x)
    =
    \frac{1}{\beta}
    \left(
        \operatorname{diag}(p(x)) - p(x)p(x)^\top
    \right).
\]
For any vector \(v\in\mathbb R^{\mathcal K}\),
\[
\begin{aligned}
    v^\top \nabla^2 F_\beta(x)v
    &=
    \frac{1}{\beta}
    \left[
        \sum_{k\in\mathcal K} p_k(x)v_k^2
        -
        \left(
            \sum_{k\in\mathcal K} p_k(x)v_k
        \right)^2
    \right].
\end{aligned}
\]
The term inside the brackets is the variance of the random variable that takes
value \(v_k\) with probability \(p_k(x)\). Hence it is nonnegative. Since
\(\beta>0\), we have
\[
    v^\top \nabla^2 F_\beta(x)v \ge 0
\]
for every \(v\in\mathbb R^{\mathcal K}\). Therefore,
\[
    \nabla^2 F_\beta(x)\succeq 0,
\]
which implies that \(F_\beta\) is convex.
\end{proof}

\begin{lemma}[Desirability Representation]
\label{lem:desirability-representation}
Fix an initial context \(q\). Let
\(
    \widehat V
    =
    \bigl(\widehat V_i(u)\bigr)_{i\in\mathcal N,\,u\succeq q}
\)
be a globally Bellman-consistent report ledger on the subtree rooted at \(q\).
Define
\[
    \widehat z_i(u)
    :=
    \exp\!\left(\frac{\widehat V_i(u)}{\beta}\right),
    \,
    \overline z(u;\widehat V)
    :=
    \sum_{j\in\mathcal N}\widehat z_j(u).
\]
Then the online decision rule is equivalently represented by
\[
    \rho_i(u;\widehat V)
    =
    \frac{\widehat z_i(u)}
    {\overline z(u;\widehat V)},
    \,
    x_{\widehat V}(a\mid u)
    =
    \pi_{\mathrm{ref}}(a\mid u)
    \frac{\overline z([u,a];\widehat V)}
    {\overline z(u;\widehat V)}
\]
for every reached non-terminal descendant \(u\succeq q\). In particular, if the
rollout terminates at \(\ell\in S_T\), then the terminal ad allocation is
\[
    \lambda_i(\ell;\widehat V)
    =
    \rho_i(\ell;\widehat V)
    =
    \frac{\widehat z_i(\ell)}
    {\overline z(\ell;\widehat V)}.
\]
\end{lemma}

\begin{proof}
Global Bellman consistency implies that, for every non-terminal descendant
\(u\succeq q\),
\[
    \widehat z_i(u)
    =
    \sum_{a\in\mathcal V}
    \pi_{\mathrm{ref}}(a\mid u)\widehat z_i([u,a]).
\]
Hence the advertiser-specific token rule can be rewritten as
\[
    \widehat\pi_i(a\mid u)
    =
    \pi_{\mathrm{ref}}(a\mid u)
    \frac{\widehat z_i([u,a])}{\widehat z_i(u)}.
\]

We prove the representation by induction along the realized path. At the root,
the initialization gives
\[
    \rho_i(q;\widehat V)
    =
    \frac{\widehat z_i(q)}
    {\overline z(q;\widehat V)}.
\]
Suppose the claim holds at a reached non-terminal state \(u\). Then
\[
\begin{aligned}
    x_{\widehat V}(a\mid u)
    &=
    \sum_{i\in\mathcal N}
    \rho_i(u;\widehat V)\widehat\pi_i(a\mid u) \\
    &=
    \sum_{i\in\mathcal N}
    \frac{\widehat z_i(u)}{\overline z(u;\widehat V)}
    \cdot
    \pi_{\mathrm{ref}}(a\mid u)
    \frac{\widehat z_i([u,a])}{\widehat z_i(u)} \\
    &=
    \pi_{\mathrm{ref}}(a\mid u)
    \frac{\overline z([u,a];\widehat V)}
    {\overline z(u;\widehat V)}.
\end{aligned}
\]
The Bayesian posterior update then gives
\[
\begin{aligned}
    \rho_i([u,a];\widehat V)
    &=
    \frac{\rho_i(u;\widehat V)\widehat\pi_i(a\mid u)}
    {x_{\widehat V}(a\mid u)}  \\
    &=
    \frac{\widehat z_i([u,a])}
    {\overline z([u,a];\widehat V)}.
\end{aligned}
\]
Thus the representation is preserved after every realized transition. The
terminal expression follows by applying the same posterior representation at
\(\ell\).
\end{proof}

\begin{lemma}[Joint Allocation Probability]
\label{lem:joint-allocation-probability}
For every globally Bellman-consistent report ledger \(\widehat V\), advertiser
\(i\in\mathcal N\), and terminal state \(\ell\in L(q)\), the joint allocation
mass assigned to advertiser \(i\) at terminal state \(\ell\) is
\[
    \alpha_{i,\ell}(\widehat V)
    :=
    \Pr_{x_{\widehat V}}(s_T=\ell\mid q)
    \lambda_i(\ell;\widehat V)
    =
    P_{0}(\ell\mid q)
    \frac{
        \exp(\widehat V_i(\ell)/\beta)
    }{
        \overline z(q;\widehat V)
    },
\]
where \(P_0(\ell\mid q)\) is the probability of reaching \(\ell\) from \(q\)
under the reference policy \(\pi_{\mathrm{ref}}\). Equivalently,
\[
    \alpha_{i,\ell}(\widehat V)
    =
    \frac{
        P_0(\ell\mid q)
        \exp(\widehat V_i(\ell)/\beta)
    }{
        \sum_{\ell'\in L(q)}
        P_0(\ell'\mid q)
        \sum_{j\in\mathcal N}
        \exp(\widehat V_j(\ell')/\beta)
    }.
\]
\end{lemma}

\begin{proof}
Let the unique path from \(q=s_0\) to \(\ell=s_T\) be
\[
    s_{t+1}=[s_t,a_t],
    \, t=0,\ldots,T-1.
\]
By Lemma~\ref{lem:desirability-representation},
\[
    x_{\widehat V}(a_t\mid s_t)
    =
    \pi_{\mathrm{ref}}(a_t\mid s_t)
    \frac{
        \overline z(s_{t+1};\widehat V)
    }{
        \overline z(s_t;\widehat V)
    }.
\]
Therefore,
\[
\begin{aligned}
    \Pr_{x_{\widehat V}}(s_T=\ell\mid q)
    &=
    \prod_{t=0}^{T-1}
    x_{\widehat V}(a_t\mid s_t) \\
    &=
    \prod_{t=0}^{T-1}
    \pi_{\mathrm{ref}}(a_t\mid s_t)
    \prod_{t=0}^{T-1}
    \frac{\overline z(s_{t+1};\widehat V)}
    {\overline z(s_t;\widehat V)} \\
    &=
    P_0(\ell\mid q)
    \frac{\overline z(\ell;\widehat V)}
    {\overline z(q;\widehat V)}.
\end{aligned}
\]
The terminal allocation equals the final posterior:
\[
    \lambda_i(\ell;\widehat V)
    =
    \rho_i(\ell;\widehat V)
    =
    \frac{\widehat z_i(\ell)}
    {\overline z(\ell;\widehat V)}
    =
    \frac{\exp(\widehat V_i(\ell)/\beta)}
    {\overline z(\ell;\widehat V)}.
\]
Multiplying the two identities gives
\[
    \alpha_{i,\ell}(\widehat V)
    =
    P_0(\ell\mid q)
    \frac{\exp(\widehat V_i(\ell)/\beta)}
    {\overline z(q;\widehat V)}.
\]

Finally, Bellman consistency implies the leaf representation
\[
    \overline z(q;\widehat V)
    =
    \sum_{j\in\mathcal N}\widehat z_j(q)
    =
    \sum_{\ell'\in L(q)}
    P_0(\ell'\mid q)
    \sum_{j\in\mathcal N}
    \exp(\widehat V_j(\ell')/\beta),
\]
which gives the equivalent normalized form.
\end{proof}

\begin{lemma}[Gradient Representation]
\label{lem:gradient-representation}
Define the soft potential
\[
    \Phi_q(\widehat V)
    :=
    \beta\log \overline z(q;\widehat V).
\]
Then the joint allocation rule is the gradient of this potential with respect to
terminal reported values:
\(
    \nabla \Phi_q(\widehat V)
    =
    \alpha(\widehat V),
\)
or equivalently, for every \(i\in\mathcal N\) and \(\ell\in L(q)\),
\(
    \frac{\partial \Phi_q(\widehat V)}
    {\partial \widehat V_i(\ell)}
    =
    \alpha_{i,\ell}(\widehat V).
\)
Consequently, \(\alpha\) is cyclically monotone.
\end{lemma}
\begin{proof}
By the leaf representation of the induced desirability functions,
\[
\begin{aligned}
    \overline z(q;\widehat V)
    &=
    \sum_{j\in\mathcal N}\widehat z_j(q) \\
    &=
    \sum_{\ell\in L(q)}
    P_0(\ell\mid q)
    \sum_{j\in\mathcal N}
    \exp\left(\frac{\widehat V_j(\ell)}{\beta}\right).
\end{aligned}
\]
Thus
\[
    \Phi_q(\widehat V)
    =
    \beta\log
    \left[
    \sum_{\ell\in L(q)}
    P_0(\ell\mid q)
    \sum_{j\in\mathcal N}
    \exp\left(\frac{\widehat V_j(\ell)}{\beta}\right)
    \right].
\]
Differentiating directly with respect to \(\widehat V_i(\ell)\) gives
\[
\begin{aligned}
    \frac{\partial \Phi_q(\widehat V)}
    {\partial \widehat V_i(\ell)}
    &=
    \beta
    \frac{
        P_0(\ell\mid q)
        \frac{1}{\beta}
        \exp\left(\widehat V_i(\ell)/\beta\right)
    }{
        \overline z(q;\widehat V)
    } \\
    &=
    P_0(\ell\mid q)
    \frac{
        \exp\left(\widehat V_i(\ell)/\beta\right)
    }{
        \overline z(q;\widehat V)
    } \\
    &=
    \alpha_{i,\ell}(\widehat V),
\end{aligned}
\]
where the last equality follows from
Lemma~\ref{lem:joint-allocation-probability}.

Finally, \(\Phi_q\) is a positive scalar multiple of a log-sum-exp function, and by Lemma~\ref{lem:logsumexp-convex}, it is therefore convex. Since \(\alpha=\nabla \Phi_q\), the allocation rule
\(\alpha\) is cyclically monotone.
\end{proof}

\begin{proposition}[Direct implementability of LAMA]
\label{prop:direct-latent-advertiser-mixture-implementable}
For any prefix \(q\), the direct latent-mixture allocation rule induced by
globally Bellman-consistent reports is implementable with quasi-linear payments.
Therefore truthful reporting \(\widehat V_i=V_i^*\) is DSIC and IR.
\end{proposition}

\begin{proof}[Proof of Proposition~\ref{prop:direct-latent-advertiser-mixture-implementable}]
Fix advertiser \(i\), true type \(r_i\), opponent reports \(\widehat V_{-i}\),
and an arbitrary globally Bellman-consistent report \(\widehat V_i\). Define the
outside-option offset
\[
    \Phi^0_{-i,q}(\widehat V_{-i})
    :=
    \beta\log
    \left(
        1+
        \sum_{j\ne i}\widehat z_j(q)
    \right).
\]
The direct payment corresponding to the pathwise LAMA payment rule is the
terminal-settled charge
\[
    p_i(\ell;\widehat V)
    :=
    \lambda_i(\ell;\widehat V)\widehat V_i(\ell)
    -
    \Phi_q(\widehat V)
    +
    \Phi^0_{-i,q}(\widehat V_{-i}).
\]
Indeed, the sequential entry fee equals
\(
    \rho_i(q;\widehat V)\widehat V_i(q)
    -\Phi_q(\widehat V)
    +\Phi^0_{-i,q}(\widehat V_{-i})
\), while the subsequent continuation-value charges telescope to
\(
    \lambda_i(\ell;\widehat V)\widehat V_i(\ell)
    -
    \lambda_i(q;\widehat V)\widehat V_i(q)
\).

Advertiser \(i\)'s expected utility under report \(\widehat V_i\) is therefore
\[
\begin{aligned}
    U_i(q;r_i,\widehat V_i,\widehat V_{-i})
    = &
    \sum_{\ell\in L(q)}
    \Pr_{x_{\widehat V}}(s_T=\ell\mid q)
    \lambda_i(\ell;\widehat V)
    \bigl(r_i(\ell)-\widehat V_i(\ell)\bigr)
    \\
    & +
    \Phi_q(\widehat V)
    -
    \Phi^0_{-i,q}(\widehat V_{-i}) \\
    = &
    \sum_{\ell\in L(q)}
    \alpha_{i,\ell}(\widehat V)
    \bigl(r_i(\ell)-\widehat V_i(\ell)\bigr)
    +
    \Phi_q(\widehat V)
    -
    \Phi^0_{-i,q}(\widehat V_{-i}).
\end{aligned}
\]
By Lemma~\ref{lem:gradient-representation},
\[
    \alpha_{i,\ell}(\widehat V)
    =
    \frac{\partial \Phi_q(\widehat V)}
    {\partial \widehat V_i(\ell)}.
\]
Hence, with the inner product taken over terminal states \(L(q)\),
\[
\begin{aligned}
    U_i(q;r_i,\widehat V_i,\widehat V_{-i})
    = &
    \Phi_q(\widehat V_i,\widehat V_{-i})
    +
    \left\langle
        \nabla_i\Phi_q(\widehat V_i,\widehat V_{-i}),
        r_i-\widehat V_i|_{L(q)}
    \right\rangle
    \\ & -
    \Phi^0_{-i,q}(\widehat V_{-i}).
\end{aligned}
\]

Since \(\Phi_q\) is convex in advertiser \(i\)'s terminal reported values when
\(\widehat V_{-i}\) is fixed, and since the truthful Bellman-consistent report
\(V_i^*\) satisfies \(V_i^*(\ell)=r_i(\ell)\) for every \(\ell\in L(q)\),
\[
    \Phi_q(V_i^*,\widehat V_{-i})
    \ge
    \Phi_q(\widehat V_i,\widehat V_{-i})
    +
    \left\langle
        \nabla_i\Phi_q(\widehat V_i,\widehat V_{-i}),
        r_i-\widehat V_i|_{L(q)}
    \right\rangle .
\]
Therefore,
\[
    U_i(q;r_i,\widehat V_i,\widehat V_{-i})
    \le
    \Phi_q(V_i^*,\widehat V_{-i})
    -
    \Phi^0_{-i,q}(\widehat V_{-i}).
\]
When \(\widehat V_i=V_i^*\), the inequality is tight, so truthful reporting is a
dominant strategy.

The same argument also gives strict dominance under the full-support
conditions. Since \(|\mathcal N|\ge2\) and \(P_0(\ell\mid q)>0\) for every
\(\ell\in L(q)\). Fixing \(\widehat V_{-i}\), write
\[
    B_{-i}(\widehat V_{-i})
    :=
    \sum_{\ell\in L(q)}
    P_0(\ell\mid q)
    \sum_{j\ne i}
    \exp(\widehat V_j(\ell)/\beta).
\]
Then \(B_{-i}(\widehat V_{-i})>0\), and as a function of advertiser \(i\)'s
terminal report \(y=(y_\ell)_{\ell\in L(q)}\),
\[
    \Phi_q(y,\widehat V_{-i})
    =
    \beta\log
    \left(
        B_{-i}(\widehat V_{-i})
        +
        \sum_{\ell\in L(q)}
        P_0(\ell\mid q)\exp(y_\ell/\beta)
    \right),
\]
Following the proof of Lemma~\ref{lem:logsumexp-convex}, its Hessian is
\(
    \frac{1}{\beta}
    \left(\operatorname{diag}(p)-pp^\top\right),
\)
where \( p_\ell=
    \frac{P_0(\ell\mid q)\exp(y_\ell/\beta)}
    {B_{-i}(\widehat V_{-i})+\sum_{\ell'}P_0(\ell'\mid q)\exp(y_{\ell'}/\beta)} \). 
For any non-zero vector $z=(z_\ell)_{\ell\in L(q)}$, we have 
\(
z^\top(\operatorname{diag}(p)-pp^\top)z = 
\sum_\ell p_\ell z_\ell^2-\left(\sum_\ell p_\ell z_\ell\right)^2.
\)
Let $s=\sum_\ell p_\ell<1,\, \bar p_\ell=\frac{p_\ell}{s}$, then

\[
\begin{aligned}
\sum_\ell p_\ell z_\ell^2-\left(\sum_\ell p_\ell z_\ell\right)^2 & = s\sum_\ell \bar p_\ell z_\ell^2 - s^2\left(\sum_\ell \bar p_\ell z_\ell\right)^2 \\
& = s\operatorname{Var}_{\bar p}(z)+s(1-s)\left(\mathbb E_{\bar p}[z]\right)^2 > 0.
\end{aligned}    
\]

Thus this Hessian is positive definite. Hence the convexity inequality above is strict
whenever \(\widehat V_i|_{L(q)}\ne r_i\). Since global Bellman consistency
uniquely determines $V^*$ from terminal values, truthful reporting
strictly dominates every distinct globally Bellman-consistent report.

Finally, truthful utility is
\[
    U_i(q;r_i,V_i^*,\widehat V_{-i})
    =
    \Phi_q(V_i^*,\widehat V_{-i})
    -
    \Phi^0_{-i,q}(\widehat V_{-i}).
\]
Because \(r_i(\ell)\ge 0\), Bellman consistency gives
\[
    z_i^*(q)
    =
    \sum_{\ell\in L(q)}
    P_0(\ell\mid q)
    \exp(r_i(\ell)/\beta)
    \ge
    1.
\]
Thus
\[
    \overline z(q;V_i^*,\widehat V_{-i})
    =
    z_i^*(q)+\sum_{j\ne i}\widehat z_j(q)
    \ge
    1+
    \sum_{j\ne i}\widehat z_j(q).
\]
Since \(\beta\log(\cdot)\) is increasing,
\[
    \Phi_q(V_i^*,\widehat V_{-i})
    \ge
    \Phi^0_{-i,q}(\widehat V_{-i}),
\]
and truthful expected utility is nonnegative.
\end{proof}

\begin{lemma}[Subtree patching]
\label{lem:subtree-patching}
Fix a reachable platform history \(h\) with current non-terminal state \(s\),
after advertiser \(i\) has reported truthfully before \(h\). Every feasible
continuation deviation \(\sigma_i'\) from \(h\) can be embedded in a globally
Bellman-consistent direct report \(\widetilde V_i\) on the full tree rooted at
\(q\): the report \(\widetilde V_i\) agrees with \(V_i^*\) outside the subtree
rooted at \(s\) and reproduces \(\sigma_i'\) inside that subtree. Holding
opponents fixed, conditional on \(h\), the patched report induces the same
continuation token distribution, terminal allocation, and future payments as
\(\sigma_i'\).
\end{lemma}

\begin{proof}[Proof of Lemma~\ref{lem:subtree-patching}]
For every non-terminal \(u\succeq s\), let
\(\widehat V_i^{\sigma'}([u,a])\) be the \(a\)-component submitted by
\(\sigma_i'\) at \(u\), and define the value at \(s\) by its soft Bellman
aggregate. Feasibility implies that these values are Bellman-consistent on the
subtree rooted at \(s\). If \(s\ne q\), it also implies
\(
    \widehat V_i^{\sigma'}(s)=V_i^*(s)
\), because the latter is the value already stored in the ledger after the
truthful prefix.

Now patch this subtree into the truthful direct report:
\[
    \widetilde V_i(u)
    :=
    \begin{cases}
        \widehat V_i^{\sigma'}(u), & u\succeq s,\\
        V_i^*(u), & u\not\succeq s.
    \end{cases}
\]
Lemma~\ref{lem:optimal-soft-bellman} gives Bellman consistency of \(V_i^*\).
When \(s\ne q\), the equality at \(s\) leaves every ancestor's Bellman equation
unchanged; when \(s=q\), there are no ancestors. Feasibility gives every
Bellman equation within the patched subtree. Hence \(\widetilde V_i\) is
globally Bellman-consistent on the tree rooted at \(q\).

Before reaching \(s\), the patched and truthful ledgers coincide. From \(s\)
onward, the patched child values are exactly those submitted by \(\sigma_i'\).
Therefore, against fixed opponent reports, the posterior at \(h\) is unchanged,
and Lemma~\ref{lem:desirability-representation} gives the same continuation
token rule and terminal allocation. The instantaneous payment rule then gives
the same payment on every continuation transition. This proves the claimed
conditional equivalence.
\end{proof}

\begin{proof}[Proof of Theorem~\ref{thm:latent-advertiser-mixture-markov-dsic}]
Lemma~\ref{lem:optimal-soft-bellman} first shows that the truthful child-value
reports are feasible. Fix advertiser \(i\), opponents' feasible strategies
\(\sigma_{-i}\), and a reachable history \(h\) with current state \(s\), after
\(i\) has reported truthfully before \(h\). Recursively recording the child
values submitted by each opponent gives a globally Bellman-consistent direct
ledger \(\widehat V_{-i}\). By
Lemma~\ref{lem:desirability-representation} and the telescoping payment rule,
the sequential mechanism is pathwise equivalent to the direct implementation
under these ledgers.

Let \(\sigma_i'\) be any feasible continuation deviation. By
Lemma~\ref{lem:subtree-patching}, it induces a globally Bellman-consistent direct
report \(\widetilde V_i\) that differs from \(V_i^*\) only on the subtree rooted
at \(s\). Proposition~\ref{prop:direct-latent-advertiser-mixture-implementable}
therefore implies
\[
\begin{aligned}
    0
    &\le
    U_i(q;r_i,V_i^*,\widehat V_{-i})
    -
    U_i(q;r_i,\widetilde V_i,\widehat V_{-i}) \\
    &=
    \Pr(h)
    \Bigl[
        U_i^\Gamma
        (h;r_i,\sigma_i^{\mathrm{tr}},\sigma_{-i};r_{-i})
        -
        U_i^\Gamma
        (h;r_i,\sigma_i',\sigma_{-i};r_{-i})
    \Bigr].
\end{aligned}
\]
Here \(\Pr(h)\) is the common probability of reaching \(h\). The equality holds
because the two direct reports coincide before \(h\) and on every branch not
passing through \(s\); past payments cancel, and conditional on \(h\) the
patched report reproduces \(\sigma_i'\). Since \(h\) is reachable,
\(\Pr(h)>0\), so the bracketed continuation-utility difference is nonnegative.
This proves Markov DSIC.

Finally, from the initial query, the truthful sequential outcome and payments
are pathwise equivalent to the direct report \((V_i^*,\widehat V_{-i})\).
The IR part of
Proposition~\ref{prop:direct-latent-advertiser-mixture-implementable} therefore
gives
\(
    U_i^\Gamma
    (q;r_i,\sigma_i^{\mathrm{tr}},\sigma_{-i};r_{-i})\ge 0
\), proving IR.
\end{proof}

\begin{proof}[Proof of Theorem~\ref{thm:latent-advertiser-mixture-efficiency}]
For every terminal state \(\ell\in L(q)\), define
\[
    G_\beta(\ell)
    :=
    \beta\log\left(
        \sum_{i\in\mathcal N}
        \exp\left(\frac{r_i(\ell)}{\beta}\right)
    \right).
\]

We first characterize the terminal allocation induced by LAMA.
Under truthful reporting, the terminal boundary condition of each
advertiser-specific value function is
\[
    V_i^*(\ell)=r_i(\ell).
\]
Therefore, by Lemma~\ref{lem:desirability-representation}, the terminal
posterior, and hence the terminal allocation, is
\[
    \lambda_i^\beta(\ell)
    =
    \frac{\exp(V_i^*(\ell)/\beta)}
         {\sum_{j\in\mathcal N}\exp(V_j^*(\ell)/\beta)}
    =
    \frac{\exp(r_i(\ell)/\beta)}
         {\sum_{j\in\mathcal N}\exp(r_j(\ell)/\beta)}.
\]

This allocation is precisely the optimizer of the entropy-regularized
terminal allocation problem. Indeed, relative to the uniform distribution on
\(\mathcal N\),
\[
    \operatorname{\mathbb{D}_{KL}}
    \bigl(\lambda\,\Vert\,U(\mathcal N)\bigr)
    =
    \log|\mathcal N|-\mathcal H(\lambda),
\]
and hence
\[
    \sum_{i\in\mathcal N}\lambda_i r_i(\ell)
    +
    \beta\mathcal H(\lambda)
    =
    \sum_{i\in\mathcal N}\lambda_i r_i(\ell)
    -
    \beta\operatorname{\mathbb{D}_{KL}}
    \bigl(\lambda\,\Vert\,U(\mathcal N)\bigr)
    +
    \beta\log|\mathcal N|.
\]
Thus the entropy-regularized objective differs by the constant
\(\beta\log|\mathcal N|\) from the KL-regularized one-step problem in
Lemma~\ref{lem:optimal-soft-bellman}, with action set \(\mathcal N\) and
terminal reward \(r_i(\ell)\) for action \(i\). Applying the lemma and adding
the same constant to its optimal value gives
\begin{equation}
\label{eq:lama-terminal-variational}
\begin{aligned}
    G_\beta(\ell)
    &=
    \max_{\lambda\in\Delta(\mathcal N)}
    \left\{
        \sum_{i\in\mathcal N}\lambda_i r_i(\ell)
        +
        \beta\mathcal H(\lambda)
    \right\},                                                     \\
    G_\beta(\ell)
    &=
    \sum_{i\in\mathcal N}\lambda_i^\beta(\ell)r_i(\ell)
    +
    \beta\mathcal H(\lambda^\beta(\ell)).
\end{aligned}
\end{equation}
The maximizer is unique because Shannon entropy is strictly concave on the
simplex.

We next characterize the token policy. Define
\[
    W_\beta(u)
    :=
    \beta\log\left(
        \sum_{i\in\mathcal N}
        \exp\left(\frac{V_i^*(u)}{\beta}\right)
    \right).
\]
At a terminal state \(\ell\), the terminal boundary condition gives
\[
    W_\beta(\ell)=G_\beta(\ell).
\]
Moreover, Lemma~\ref{lem:optimal-soft-bellman} implies, at every non-terminal
state \(u\),
\begin{equation}
\label{eq:lama-aggregate-bellman}
    \exp\left(\frac{W_\beta(u)}{\beta}\right)
    =
    \sum_{a\in\mathcal V}
    \pi_{\mathrm{ref}}(a\mid u)
    \exp\left(\frac{W_\beta([u,a])}{\beta}\right).
\end{equation}
Indeed, this follows by exponentiating each advertiser-specific soft Bellman
recursion and then summing over advertisers. By
Lemma~\ref{lem:desirability-representation}, the truthful LAMA token policy is
\[
    x^\beta(a\mid u)
    =
    \frac{
        \pi_{\mathrm{ref}}(a\mid u)
        \exp(W_\beta([u,a])/\beta)
    }{
        \sum_{a'\in\mathcal V}
        \pi_{\mathrm{ref}}(a'\mid u)
        \exp(W_\beta([u,a'])/\beta)
    }.
\]
Together with the terminal boundary condition,
Equation~\eqref{eq:lama-aggregate-bellman} identifies \(W_\beta\) as the
optimal soft value function for terminal reward \(G_\beta\). Hence
Lemma~\ref{lem:optimal-soft-bellman} gives
\begin{equation}
\label{eq:lama-relaxed-token-control}
    x^\beta
    \in
    \arg\max_x
    \mathbb E_{\tau\sim x(\cdot\mid q)}
    \left[
        G_\beta(\ell(\tau))
        -
        \beta
        \sum_{t=0}^{T_\tau-1}
        \log
        \frac{x(a_t\mid s_t)}
             {\pi_{\mathrm{ref}}(a_t\mid s_t)}
    \right].
\end{equation}

Combining the terminal optimization with the outer token-control
problem proves the first claim. Indeed, for any fixed token policy
\(x\), the choice of \(\lambda(\tau)\) is separable across terminal
trajectories. By Equation~\eqref{eq:lama-terminal-variational},
\[
\begin{aligned}
    &\max_{\lambda}
    \mathbb E_{\tau\sim x(\cdot\mid q)}
    \left[
    \begin{aligned}
        &\sum_{i\in\mathcal N}\lambda_i(\tau)r_i(\ell(\tau))
        +
        \beta\mathcal H(\lambda(\tau))
        \\
        &\quad-
        \beta\sum_{t=0}^{T_\tau-1}
        \log
        \frac{x(a_t\mid s_t)}
             {\pi_{\mathrm{ref}}(a_t\mid s_t)}
    \end{aligned}
    \right]                                                       \\
    &\qquad =
    \mathbb E_{\tau\sim x(\cdot\mid q)}
    \left[
        G_\beta(\ell(\tau))
        -
        \beta\sum_{t=0}^{T_\tau-1}
        \log
        \frac{x(a_t\mid s_t)}
             {\pi_{\mathrm{ref}}(a_t\mid s_t)}
    \right].
\end{aligned}
\]
The inner maximizer is \(\lambda^\beta\), while
Equation~\eqref{eq:lama-relaxed-token-control} shows that the outer maximizer is
\(x^\beta\). Hence
\((x^\beta,\lambda^\beta)\) maximizes the stated entropy-regularized
welfare.

We now prove the welfare bound. For every
\(\lambda\in\Delta(\mathcal N)\), each term
\(-\lambda_i\log\lambda_i\) is nonnegative, while
\[
    \log|\mathcal N|-\mathcal H(\lambda)
    =
    \operatorname{\mathbb{D}_{KL}}
    \bigl(\lambda\,\Vert\, U(\mathcal N)\bigr)
    \geq 0.
\]
Therefore, in particular,
\[
    0
    \leq
    \mathcal H(\lambda^\beta(\ell))
    \leq
    \log|\mathcal N|.
\]
Moreover, for every \(i\in\mathcal N\),
\[
\begin{aligned}
    G_\beta(\ell)
    &=
    \beta\log
    \sum_{j\in\mathcal N}
    \exp\left(\frac{r_j(\ell)}{\beta}\right)                    \\
    &\geq
    \beta\log
    \exp\left(\frac{r_i(\ell)}{\beta}\right)
    =
    r_i(\ell).
\end{aligned}
\]
Taking the maximum over \(i\) gives
\(G_\beta(\ell)\geq\max_{i\in\mathcal N}r_i(\ell)\). Using
Equations~\eqref{eq:lama-terminal-variational} and
\eqref{eq:lama-relaxed-token-control}, we obtain
\[
\begin{aligned}
    &\mathrm{SW}_\beta
        (x^\beta,\lambda^\beta\mid q,r)                          \\
    = &
    \sup_x
    \mathbb E_{\tau\sim x(\cdot\mid q)}
    \left[
        G_\beta(\ell(\tau))
        -
        \beta\sum_{t=0}^{T_\tau-1}
        \log
        \frac{x(a_t\mid s_t)}
             {\pi_{\mathrm{ref}}(a_t\mid s_t)}
    \right]
    \\ &-
    \beta\,
    \mathbb E_{\tau\sim x^\beta(\cdot\mid q)}
    \left[\mathcal H(\lambda^\beta(\tau))\right]                \\
    \geq &
    \sup_x
    \mathbb E_{\tau\sim x(\cdot\mid q)}
    \left[
        G_\beta(\ell(\tau))
        -
        \beta\sum_{t=0}^{T_\tau-1}
        \log
        \frac{x(a_t\mid s_t)}
             {\pi_{\mathrm{ref}}(a_t\mid s_t)}
    \right]
    \\ &-
    \beta\log|\mathcal N|                                      \\
    \geq &
    \sup_x
    \mathbb E_{\tau\sim x(\cdot\mid q)}
    \left[
        \max_{i\in\mathcal N}r_i(\ell(\tau))
        -
        \beta\sum_{t=0}^{T_\tau-1}
        \log
        \frac{x(a_t\mid s_t)}
             {\pi_{\mathrm{ref}}(a_t\mid s_t)}
    \right]
    \\ &-
    \beta\log|\mathcal N|                                      \\
    = &
    \sup_{\pi,\lambda}
    \mathrm{SW}_\beta(\pi,\lambda\mid q,r)
    -
    \beta\log|\mathcal N|.
\end{aligned}
\]
The last equality follows because, conditional on each terminal state, the
simplex first-best assigns the opportunity to an advertiser with maximal
value. Since
\((x^\beta,\lambda^\beta)\) is itself feasible for the original
simplex-allocation problem,
\[
    \mathrm{SW}_\beta
        (x^\beta,\lambda^\beta\mid q,r)
    \leq
    \sup_{\pi,\lambda}
    \mathrm{SW}_\beta(\pi,\lambda\mid q,r).
\]
Combining the last two displays gives
\[
    0
    \leq
    \sup_{\pi,\lambda}
    \mathrm{SW}_\beta(\pi,\lambda\mid q,r)
    -
    \mathrm{SW}_\beta
        (x^\beta,\lambda^\beta\mid q,r)
    \leq
    \beta\log|\mathcal N|.
\]

Finally, let \(\ell^*\) be a maximizing terminal state satisfying
\(P_0(\ell^*\mid q)>0\). Applying
Lemma~\ref{lem:optimal-soft-bellman} to the terminal reward
\(\max_{i\in\mathcal N}r_i(\ell)\) gives
\[
    \sup_{\pi,\lambda}
    \mathrm{SW}_\beta(\pi,\lambda\mid q,r)
    =
    \beta\log
    \sum_{\ell\in L(q)}
    P_0(\ell\mid q)
    \exp\left(
        \frac{\max_{i\in\mathcal N}r_i(\ell)}{\beta}
    \right).
\]
It follows that
\[
\begin{aligned}
    & \max_{\ell\in L(q)}\max_{i\in\mathcal N}r_i(\ell)
    +
    \beta\log P_0(\ell^*\mid q)
    \\ \leq&
    \sup_{\pi,\lambda}
    \mathrm{SW}_\beta(\pi,\lambda\mid q,r)
    \leq
    \max_{\ell\in L(q)}\max_{i\in\mathcal N}r_i(\ell).
\end{aligned}
\]
Both bounds converge to
\(\max_{\ell\in L(q)}\max_{i\in\mathcal N}r_i(\ell)\) as
\(\beta\to0^+\). Together with the welfare-gap bound, the squeeze theorem
yields
\[
    \mathrm{SW}_\beta
        (x^\beta,\lambda^\beta\mid q,r)
    \to
    \max_{\ell\in L(q)}\max_{i\in\mathcal N}r_i(\ell).
\]
The welfare-gap inequality is pointwise in \(q\), so taking expectation
over \(q\sim\mathcal D\) preserves the same
\(\beta\log|\mathcal N|\) bound.
\end{proof}

\begin{proof}[Proof of Theorem~\ref{thm:value-report-consistency}]
Fix an advertiser-query pair $(i,q)$, and let $\pi^\dagger$ be any population
minimizer of the local-advantage objective. Define
\[
    \widehat A_i^\dagger(s,a)
    :=
    \beta\log
    \frac{\pi_i^\dagger(a\mid s)}
         {\pi_{\mathrm{ref}}(a\mid s)}.
\]
For a completed response
$y=(a_0,\ldots,a_{T_y-1})$, where
$s_0=q$, $s_{t+1}=[s_t,a_t]$, and $s_{T_y}=\ell_y$, define its
cumulative score
\[
    G_i^\dagger(q,y)
    :=
    \sum_{t=0}^{T_y-1}
    \widehat A_i^\dagger(s_t,a_t).
\]
By autoregressive factorization,
\[
\begin{aligned}
    G_i^\dagger(q,y)
    &=
    \sum_{t=0}^{T_y-1}
    \beta\log
    \frac{\pi_i^\dagger(a_t\mid s_t)}
         {\pi_{\mathrm{ref}}(a_t\mid s_t)}
\\
    &=
    \beta\log
    \frac{\pi_i^\dagger(y\mid q)}
         {\pi_{\mathrm{ref}}(y\mid q)}.
\end{aligned}
\]
For a comparison pair $(y,y')$, let
\[
    \Delta_i^\dagger(q;y,y')
    :=
    G_i^\dagger(q,y)-G_i^\dagger(q,y').
\]
The Bradley-Terry target is
\[
    \omega_i(q;y,y')
    =
    \sigma\left(
        r_i(\ell_y)-r_i(\ell_{y'})
    \right).
\]
For a predicted logit $\Delta$, abbreviate
\(
    p:=\omega_i(q;y,y')
\)
and
\(
    \widehat p_\Delta:=\sigma(\Delta).
\)
Since \(\sigma(-\Delta)=1-\widehat p_\Delta\), the binary cross-entropy satisfies
\[
\begin{aligned}
&-p\log \widehat p_\Delta-(1-p)\log(1-\widehat p_\Delta)
\\
&=
\bigl[-p\log p-(1-p)\log(1-p)\bigr]
\\
&\quad+
p\log\frac{p}{\widehat p_\Delta}
+(1-p)\log\frac{1-p}{1-\widehat p_\Delta}
\\
&=
\mathcal H\bigl(\operatorname{Bern}(p)\bigr)
+
\operatorname{\mathbb{D}_{KL}}
\bigl(
    \operatorname{Bern}(p)
    \,\Vert\,
    \operatorname{Bern}(\widehat p_\Delta)
\bigr).
\end{aligned}
\]
The middle equality adds and subtracts the Bernoulli entropy terms. By
nonnegativity of the KL divergence, the loss is minimized exactly when
\(\widehat p_\Delta=p\).

By Lemma~\ref{lem:optimal-soft-bellman}, the optimal policy satisfies
\[
    \beta\log
    \frac{\pi_i^*(a\mid s)}
         {\pi_{\mathrm{ref}}(a\mid s)}
    =
    V_i^*([s,a])-V_i^*(s).
\]
Telescoping along a completed response gives
\[
    \beta\log
    \frac{\pi_i^*(y\mid q)}
         {\pi_{\mathrm{ref}}(y\mid q)}
    =
    r_i(\ell_y)-V_i^*(q).
\]
Hence \(\pi_i^*\) induces the target logit
\(r_i(\ell_y)-r_i(\ell_{y'})\). Because \(\pi_i^*\) is feasible, the pointwise
lower bound is attainable, so every population minimizer satisfies
\[
    \sigma\left(\Delta_i^\dagger(q;y,y')\right)
    =
    \sigma\left(
        r_i(\ell_y)-r_i(\ell_{y'})
    \right)
\]
for every pair of completed responses. Since \(\sigma\) is injective,
\[
    G_i^\dagger(q,y)-G_i^\dagger(q,y')
    =
    r_i(\ell_y)-r_i(\ell_{y'}).
\]
Therefore there exists a constant $c_i(q)$ such that
\begin{equation}
\label{eq:exact-recovery-path-score}
    G_i^\dagger(q,y)
    =
    r_i(\ell_y)+c_i(q).
\end{equation}

Equation~\eqref{eq:exact-recovery-path-score} identifies cumulative scores up
to the query-specific constant \(c_i(q)\). We next determine this constant from
token-policy normalization.

Normalization of \(\pi_i^\dagger\) implies, at every reachable non-terminal
prefix $s$,
\begin{equation}
\label{eq:exact-recovery-normalization}
    \sum_{a\in\mathcal V}
    \pi_{\mathrm{ref}}(a\mid s)
    \exp\left(
        \frac{\widehat A_i^\dagger(s,a)}{\beta}
    \right)
    =
    \sum_{a\in\mathcal V}\pi_i^\dagger(a\mid s)
    =1.
\end{equation}
Because token states are prefixes, every reachable state has a unique
path from the root. For $s_t=[q,a_{<t}]$, define
\begin{equation}
\label{eq:exact-recovery-candidate-value}
\begin{aligned}
    \widetilde V_i(s_t)
    &:=
    -c_i(q)
    +
    \sum_{k=0}^{t-1}
    \widehat A_i^\dagger(s_k,a_k),
\\
    \widetilde V_i([s,a])
    &:=
    \widetilde V_i(s)+\widehat A_i^\dagger(s,a).
\end{aligned}
\end{equation}
At a terminal state $\ell_y$, Equation~\eqref{eq:exact-recovery-path-score}
gives
\[
    \widetilde V_i(\ell_y)
    =
    -c_i(q)+G_i^\dagger(q,y)
    =
    r_i(\ell_y).
\]
For every non-terminal prefix $s$,
Equations~\eqref{eq:exact-recovery-normalization}
and~\eqref{eq:exact-recovery-candidate-value} imply
\[
\begin{aligned}
    &\beta\log
    \sum_{a\in\mathcal V}
    \pi_{\mathrm{ref}}(a\mid s)
    \exp\left(
        \frac{\widetilde V_i([s,a])}{\beta}
    \right)
\\
    &=
    \widetilde V_i(s)
    +
    \beta\log
    \sum_{a\in\mathcal V}
    \pi_{\mathrm{ref}}(a\mid s)
    \exp\left(
        \frac{\widehat A_i^\dagger(s,a)}{\beta}
    \right)
    =
    \widetilde V_i(s).
\end{aligned}
\]
Thus \(\widetilde V_i\) satisfies the soft Bellman system in
Lemma~\ref{lem:optimal-soft-bellman}. Its finite-horizon solution is unique by
backward induction from the terminal boundary, so
\[
    \widetilde V_i(s)=V_i^*(s)
\]
at every reachable state. In particular,
\begin{equation}
\label{eq:exact-recovery-root-constant}
    -c_i(q)=\widetilde V_i(q)=V_i^*(q).
\end{equation}

We can now identify the local scores. For a fixed pair \((s,a)\), write a
continuation trajectory from the child state as
\(\tau=(u_0,b_0,\ldots,u_{T_\tau})\), where \(u_0=[s,a]\). The optimal
soft advantage \(A_i^*(s,a)\)of committing to \(a\) at \(s\) is commonly defined as
\[
\begin{aligned}
    &\sup_{\pi}
    \mathbb E_{\tau\sim\pi(\cdot\mid u_0)}
    \left[
        r_i(u_{T_\tau})
        -
        \beta
        \sum_{k=0}^{T_\tau-1}
        \log
        \frac{\pi(b_k\mid u_k)}
             {\pi_{\mathrm{ref}}(b_k\mid u_k)}
    \right]
    -V_i^*(s).
\end{aligned}
\]
Because the transition is deterministic and there are no intermediate rewards,
the supremum is \(V_i^*([s,a])\). Thus
\(A_i^*(s,a)=V_i^*([s,a])-V_i^*(s)\), and
Equation~\eqref{eq:exact-recovery-candidate-value} gives
\begin{equation}
\label{eq:exact-recovery-advantage}
\begin{aligned}
    \widehat A_i^\dagger(s,a)
    =
    \widetilde V_i([s,a])-\widetilde V_i(s)
    =
    V_i^*([s,a])-V_i^*(s)
    =
    A_i^*(s,a).
\end{aligned}
\end{equation}
Using the policy formula in Lemma~\ref{lem:optimal-soft-bellman}, the definition
of \(\widehat A_i^\dagger\), and
Equation~\eqref{eq:exact-recovery-advantage},
\[
\begin{aligned}
    \pi_i^\dagger(a\mid s)
    &=
    \pi_{\mathrm{ref}}(a\mid s)
    \exp\left(
        \frac{\widehat A_i^\dagger(s,a)}{\beta}
    \right)
\\
    &=
    \pi_{\mathrm{ref}}(a\mid s)
    \exp\left(
        \frac{A_i^*(s,a)}{\beta}
    \right)
\\
    &=
    \pi_i^*(a\mid s).
\end{aligned}
\]
Thus \(\pi_i^\dagger=\pi_i^*\). Since \((i,q)\) was arbitrary, the population
minimizer is unique in induced policy space.

It remains to identify the root value used to initialize the ledger. Combining
Equations~\eqref{eq:exact-recovery-path-score}
and~\eqref{eq:exact-recovery-root-constant} gives, for every completed response
$y$,
\[
\begin{aligned}
    r_i(\ell_y)
    -
    \sum_{t=0}^{T_y-1}
    \widehat A_i^\dagger(s_t,a_t)
    &=
    r_i(\ell_y)-G_i^\dagger(q,y)
    =
    -c_i(q)
    =
    V_i^*(q).
\end{aligned}
\]
The residual is pointwise equal to \(V_i^*(q)\). Taking its conditional
expectation therefore gives
\[
    v_i^\dagger(q)
    =
    \mathbb E_{y\sim\mathcal D(\cdot\mid i,q)}
    \left[
        r_i(\ell_y)
        -
        \sum_{t=0}^{T_y-1}
        \widehat A_i^\dagger(s_t,a_t)
    \right]
    =
    V_i^*(q).
\]
Thus the population root value is recovered on the training support.

With the root value identified, initialize the serving-time ledger by
\[
    \widehat V_i(q):=v_i^\dagger(q)=V_i^*(q)
\]
and update it recursively according to
\[
    \widehat V_i([s,a])
    :=
    \widehat V_i(s)+\widehat A_i^\dagger(s,a).
\]
If $\widehat V_i(s)=V_i^*(s)$, then
\[
\begin{aligned}
    \widehat V_i([s,a])
    & =
    \widehat V_i(s)+\widehat A_i^\dagger(s,a)
    \\ &=
    V_i^*(s)+V_i^*([s,a])-V_i^*(s)
    =
    V_i^*([s,a]).
\end{aligned}
\]
Induction over prefix length gives \(\widehat V_i(s)=V_i^*(s)\) at every
reachable state, so the constructed child-value vector is the optimal
soft-value report required by LAMA.
\end{proof}

\begin{proof}[Proof of Proposition~\ref{prop:lama-winner-payment-ir-wbb}]
Fix a query \(q\) and candidate advertiser set \(\mathcal N\). Under exact
reporting, specialize the desirability notation of
Lemma~\ref{lem:desirability-representation} to the truthful value profile:
\[
    z_i^*(q)
    :=
    \exp\left(\frac{V_i^*(q)}{\beta}\right),
    \,
    \overline z(q;V^*)
    =
    \sum_{j\in\mathcal N}z_j^*(q).
\]
The initial allocation posterior therefore satisfies
\[
    \rho_i(q)
    =
    \frac{z_i^*(q)}{\overline z(q;V^*)}.
\]

For a realized terminal state \(s_T\), let
\[
    P_i^{\mathrm{frac}}(s_T)
    :=
    \rho_i(s_T)r_i(s_T)
    -
    \beta\log\overline z(q;V^*)
    +
    \beta\log\bigl(1+\overline z(q;V^*)-z_i^*(q)\bigr)
\]
denote advertiser \(i\)'s telescoped fractional payment before the
winner-probability adjustment in Algorithm~\ref{alg:lama}. The winner-settled
payment is therefore
\[
    P_i(s_T,I^*)
    =
    \mathbf 1\{I^*=i\}
    \frac{P_i^{\mathrm{frac}}(s_T)}{\rho_i(s_T)},
    \,
    I^*\sim\rho(s_T).
\]

We first prove outcome individual rationality. Because terminal values
are nonnegative, Lemma~\ref{lem:optimal-soft-bellman} gives
\[
\begin{aligned}
    z_i^*(q)
    &=
    \exp\left(\frac{V_i^*(q)}{\beta}\right)                         \\
    &=
    \sum_{\ell\in L(q)}
    P_0(\ell\mid q)
    \exp\left(\frac{r_i(\ell)}{\beta}\right)
    \geq 1.
\end{aligned}
\]
Consequently,
\(
    \overline z(q;V^*)
    \geq
    1+\overline z(q;V^*)-z_i^*(q).
\)

Suppose advertiser \(i\) is the sampled winner at terminal state
\(s_T\). By exactness of the terminal report,
\(V_i^*(s_T)=r_i(s_T)\), so its realized utility is
\[
\begin{aligned}
    r_i(s_T)-P_i(s_T,i)
    &=
    \frac{\beta}{\rho_i(s_T)}
    \log\left(
        \frac{
            \overline z(q;V^*)
        }{
            1+\overline z(q;V^*)-z_i^*(q)
        }
    \right)
    \geq 0,
\end{aligned}
\]
where the inequality follows from the preceding desirability bound. Thus the
winner's payment does not exceed its realized gain.
If advertiser \(i\) is not selected, it receives no advertising
opportunity and pays zero, so its realized utility is also zero.
Therefore outcome IR holds pointwise for every realized terminal state
and every sampled winner.

We next prove expected weak budget balance. Conditional on a realized
terminal state \(s_T\), taking expectation only over the sampled winner
gives
\[
\begin{aligned}
    \mathbb E_{I^*\sim\rho(s_T)}
    \left[
        P_i(s_T,I^*)
        \mid s_T
    \right]
    &=
    \rho_i(s_T)
    \frac{P_i^{\mathrm{frac}}(s_T)}{\rho_i(s_T)}
    =
    P_i^{\mathrm{frac}}(s_T).
\end{aligned}
\]
Thus it suffices to show that the expected fractional payments are
nonnegative.

By the definition of \(P_i^{\mathrm{frac}}\),
\[
\begin{aligned}
    \frac{1}{\beta}
    \mathbb E_{s_T\sim x^\beta}
    \left[P_i^{\mathrm{frac}}(s_T)\right]
    = &
    \frac{1}{\beta}
    \mathbb E_{s_T\sim x^\beta}
    \left[\rho_i(s_T)r_i(s_T)\right]
    -
    \log\overline z(q;V^*)
    \\ & +
    \log\bigl(1+\overline z(q;V^*)-z_i^*(q)\bigr).
\end{aligned}
\]
The last two terms are deterministic, so the only random term that must be
controlled is the expected allocation-weighted value. By
Lemma~\ref{lem:joint-allocation-probability},
\[
\begin{aligned}
    &\mathbb E_{s_T\sim x^\beta}
    \left[\rho_i(s_T)r_i(s_T)\right]                              \\
    &\qquad =
    \frac{1}{\overline z(q;V^*)}
    \sum_{\ell\in L(q)}
    P_0(\ell\mid q)
    \exp\left(\frac{r_i(\ell)}{\beta}\right)
    r_i(\ell).
\end{aligned}
\]

To normalize this sum, define the terminal distribution
\[
    \nu_i(\ell)
    :=
    \frac{
        P_0(\ell\mid q)\exp(r_i(\ell)/\beta)
    }{
        z_i^*(q)
    }.
\]
This is a probability distribution because the leaf representation above gives
\(
    z_i^*(q)
    =
    \sum_{\ell\in L(q)}
    P_0(\ell\mid q)\exp(r_i(\ell)/\beta).
\)
The expected allocation-weighted value can now be written as
\[
\begin{aligned}
    \mathbb E_{s_T\sim x^\beta}
    \left[
        \rho_i(s_T)r_i(s_T)
    \right]
    &=
    \frac{z_i^*(q)}{\overline z(q;V^*)}
    \mathbb E_{\ell\sim\nu_i}[r_i(\ell)].
\end{aligned}
\]

It remains to lower-bound the expectation under \(\nu_i\). For every \(\ell\)
with \(\nu_i(\ell)>0\), the definition of \(\nu_i\) implies
\[
    \log
    \frac{\nu_i(\ell)}{P_0(\ell\mid q)}
    =
    \frac{r_i(\ell)}{\beta}
    -
    \log z_i^*(q).
\]
Taking expectation under \(\nu_i\) and rearranging gives
\[
    \mathbb E_{\ell\sim\nu_i}[r_i(\ell)]
    -
    \beta\operatorname{\mathbb{D}_{KL}}
    \bigl(\nu_i\Vert P_0(\cdot\mid q)\bigr)
    =
    \beta\log z_i^*(q)
    =
    V_i^*(q).
\]
Since the KL divergence is nonnegative,
\[
    \mathbb E_{\ell\sim\nu_i}[r_i(\ell)]
    \geq
    V_i^*(q)
    =
    \beta\log z_i^*(q).
\]
Combining the last inequality with the representation of the
allocation-weighted value yields
\[
    \mathbb E_{s_T\sim x^\beta}
    \left[\rho_i(s_T)r_i(s_T)\right]
    \geq
    \beta
    \frac{z_i^*(q)}{\overline z(q;V^*)}
    \log z_i^*(q).
\]
Substituting this bound into the payment decomposition gives
\[
\begin{aligned}
    \frac{1}{\beta}
    \mathbb E[P_i^{\mathrm{frac}}(s_T)]
    \geq &
    \frac{z_i^*(q)}{\overline z(q;V^*)}\log z_i^*(q)
    -
    \log\overline z(q;V^*)
    \\ &+
    \log\bigl(1+\overline z(q;V^*)-z_i^*(q)\bigr).
\end{aligned}
\]

It remains to show that the right-hand side is nonnegative. Since
\(
    0<z_i^*(q)/\overline z(q;V^*)\leq 1,
\)
the weighted arithmic-geometric mean inequality gives
\[
\begin{aligned}
    \frac{1+\overline z(q;V^*)-z_i^*(q)}{\overline z(q;V^*)}
    &=
    \left(
        1-\frac{z_i^*(q)}{\overline z(q;V^*)}
    \right)\!\cdot 1
    +
    \frac{z_i^*(q)}{\overline z(q;V^*)}
    \cdot
    \frac{1}{z_i^*(q)}                                           \\
    &\geq
    1^{\,1-z_i^*(q)/\overline z(q;V^*)}
    \left(\frac{1}{z_i^*(q)}\right)^{
        z_i^*(q)/\overline z(q;V^*)
    }
    \\ &=
    z_i^*(q)^{-z_i^*(q)/\overline z(q;V^*)}.
\end{aligned}
\]
Taking logarithms,
\[
    \log\bigl(1+\overline z(q;V^*)-z_i^*(q)\bigr)
    -
    \log\overline z(q;V^*)
    \geq
    -
    \frac{z_i^*(q)}{\overline z(q;V^*)}
    \log z_i^*(q).
\]
Therefore,
\(
    \mathbb E[P_i^{\mathrm{frac}}(s_T)]\geq 0 \)
for every
\(i\in\mathcal N.
\)

Finally, taking expectations in the conditional payment identity above and
summing over advertisers gives
\[
\begin{aligned}
    \mathbb E_{s_T,I^*}
    \left[
        \sum_{i\in\mathcal N}P_i(s_T,I^*)
    \right]
    &=
    \sum_{i\in\mathcal N}
    \mathbb E_{s_T}[P_i^{\mathrm{frac}}(s_T)]                     \\
    &\geq 0.
\end{aligned}
\]
The expectation is over both token-generation randomness and the
terminal winner draw, conditional on the query and candidate
advertiser set. Hence the winner-settled implementation is weakly
budget balanced in expectation.
\end{proof}

%% file: sections/vertical_advertiser_overview.tex
This appendix provides supporting experimental details and diagnostics. Table~\ref{tab:vertical_advertiser_overview} summarizes the evaluated verticals and their advertisers, and Tables~\ref{tab:case_study_workout}--\ref{tab:case_study_car} present representative generated responses for each vertical. Tables~\ref{tab:usedcars-kbb-carfax} and~\ref{tab:usedcars-kbb-edmunds} illustrate how the generated response changes with bids for complementary and competing advertiser pairs, respectively. Table~\ref{tab:latency-ratio} compares the end-to-end latency of LAMA with reference-answer generation. Table~\ref{tab:ctr-calibration} and Figure~\ref{fig:value-model-calibration} report the prediction and calibration errors of the learned shared report model, while Figure~\ref{fig:bid-offset-sweep} empirically evaluates its monotonicity under bid perturbations.

\begin{table*}[htbp]
\centering
\small
\renewcommand{\arraystretch}{1.3}
\begin{tabularx}{\textwidth}{l c X l l X}
\toprule
\textbf{Vertical} & \textbf{\# Queries} & \textbf{Description} & \textbf{Advertisers} & \textbf{Type} & \textbf{Qualities} \\
\midrule
\multirow{3}{*}{\textbf{Workout}} & \multirow{3}{*}{291} & \multirow{3}{*}{\parbox{\linewidth}{Fitness and workout-intent search queries, covering supplements, gym routines, class discovery, and home exercise equipment.}}
& Bowflex SelectTech 552 & Adjustable dumbbells & Replace a full rack of weights with adjustable dumbbells for home strength training. \\
\cline{4-6}
& & & C4 Sport & Pre-workout supplement & Prepare for demanding gym sessions with a familiar pre-workout option. \\
\cline{4-6}
& & & ClassPass & Gym and studio marketplace & Book gyms, barre, Pilates, boxing, and cycling classes with flexible choices. \\
\midrule
\multirow{3}{*}{\textbf{Vacation}} & \multirow{3}{*}{456} & \multirow{3}{*}{\parbox{\linewidth}{Travel-planning search queries, including destinations, packages, flights, hotels, cruises, resorts, and itinerary planning.}}
& Expedia & Travel booking platform & Compare flights, hotels, cars, and vacation packages for flexible trip planning. \\
\cline{4-6}
& & & Priceline & Travel deal booking platform & Search travel deals across hotels, flights, rental cars, and vacation packages. \\
\cline{4-6}
& & & Tripadvisor & Travel review and planning platform & Compare destinations, resorts, attractions, and traveler reviews while planning a trip. \\
\midrule
\multirow{3}{*}{\textbf{Car}} & \multirow{3}{*}{492} & \multirow{3}{*}{\parbox{\linewidth}{Automotive search queries for vehicle research, used-car shopping, pricing/value checks, EV/hybrid comparison, and dealer-local intents.}}
& Kelley Blue Book & Vehicle valuation and pricing research service & Check vehicle value ranges, pricing guidance, and market information when comparing cars. \\
\cline{4-6}
& & & Edmunds & Vehicle research and review platform & Compare vehicle reviews, specs, pricing guidance, and ownership factors when researching cars. \\
\cline{4-6}
& & & CARFAX & Vehicle history service & Review vehicle history details when evaluating used cars or comparing specific vehicles. \\
\bottomrule
\end{tabularx}
\caption{Vertical overview and advertiser metadata used in the experiments.}
\label{tab:vertical_advertiser_overview}
\end{table*}

\begin{table*}[htbp]
\centering
\small
\renewcommand{\arraystretch}{1.3}
\begin{tabularx}{\textwidth}{l X}
\toprule
\textbf{Field} & \textbf{Content} \\
\midrule
\textbf{Vertical} & Workout \\
\midrule
\textbf{Query} & best bicep exercises for mass \\
\midrule
\textbf{Reference} & The best bicep exercises for mass typically include compound and isolation movements that target the biceps effectively. Common options are barbell curls, dumbbell curls, chin-ups, hammer curls, and concentration curls. These exercises vary in intensity and muscle engagement, allowing for progressive overload, which is essential for building mass. Key factors to consider include proper form, controlled movements, and sufficient resistance to challenge the muscles. It's also important to incorporate a variety of exercises to stimulate different parts of the biceps and avoid plateaus. Training frequency, rest periods, and nutrition play crucial roles in muscle growth. For optimal results, pair these exercises with a well-rounded upper body routine and adequate protein intake. Always ensure that exercises are performed with correct technique to prevent injury, and consider consulting a fitness professional for personalized guidance. \\
\midrule
\textbf{LAMA} & The best bicep exercises for mass focus on progressive overload and proper form to effectively build size and strength. Common effective moves include barbell curls, dumbbell curls, cable chin-ups or pull-downs, hammer curls, and concentration curls. These exercises target the biceps brachii from different angles, promoting muscle hypertrophy. Key factors to prioritize are using heavy enough weights for 6-12 reps, maintaining a full range of motion, and ensuring rest and nutrition support growth. \winnerhl{For those with limited space, \textbf{@Bowflex SelectTech 552@} can be an excellent choice as they offer adjustable weights that make performing a variety of bicep-focused lifts more convenient at home.} Incorporating a mix of free weights and machine-based exercises can lead to more balanced development. Also, supersetting bicep workouts with opposing tricep movements may enhance muscle growth. Always warm up adequately to reduce injury risk and maximize performance. \\
\midrule
\textbf{Winner} & Bowflex \\
\bottomrule
\end{tabularx}
\caption{Case Study 1: Workout vertical. In the generated content, the winning advertiser is shown in bold, and the highlighted sentence is the span allocated to the hyperlink.}
\label{tab:case_study_workout}
\end{table*}

\clearpage

\begin{table*}[ht]
\centering
\small
\renewcommand{\arraystretch}{1.3}
\begin{tabularx}{\textwidth}{l X}
\toprule
\textbf{Field} & \textbf{Content} \\
\midrule
\textbf{Vertical} & Vacation \\
\midrule
\textbf{Query} & best beach vacations in the world \\
\midrule
\textbf{Reference} & The best beach vacations in the world offer a mix of pristine sands, clear waters, and unique cultural experiences, appealing to travelers seeking relaxation or adventure. Popular destinations include the Maldives for luxury overwater bungalows, Bali for vibrant culture and natural beauty, and the Caribbean for family-friendly resorts and vibrant nightlife. Key factors to consider include travel timing, with peak seasons often running from late spring to early fall, and budget, which can vary greatly depending on whether you opt for all-inclusive resorts, boutique hotels, or local guesthouses. Location preferences also play a role, with some travelers favoring secluded spots and others prioritizing proximity to activities or nightlife. Flexibility in travel dates can lead to better deals and availability, while last-minute bookings may require more adaptable options. Practical considerations include researching visa requirements, health advisories, and weather patterns, and ensuring that travel insurance covers potential disruptions. Whether planning a romantic getaway, family trip, or solo adventure, comparing package deals, rental options, and local experiences can help tailor the perfect beach vacation to individual needs and constraints. \\
\midrule
\textbf{LAMA} & The best beach vacations in the world offer a mix of stunning scenery, clear waters, and diverse activities tailored to different preferences. \texttt{@Expedia@} highlights top options like the Maldives for luxury overwater bungalows, Bali for cultural richness and vibrant nightlife, and the Caribbean for family-friendly resorts and cruise access. White-sand beaches in Greece, such as those in Santorini or Mykonos, appeal to couples seeking romance, while Mexico’s Riviera Maya offers adventure and world-class resorts. Key factors include seasonality, with peak times often from June to August, affecting budget and availability. Travelers should consider visa requirements, health advisories, and local conditions. For flexibility, all-inclusive resorts or cruise packages may suit families or solo travelers, while last-minute options abound through online travel agencies. Beach vacationers should also plan for sun protection, water activities, and transportation to remote islands. \winnerhl{\textbf{@Tripadvisor@} can help compare destinations, read traveler reviews, and find the best accommodations and experiences that match personal preferences and travel styles.} Always verify current travel advisories and check with providers for the latest updates on services and restrictions. \\
\midrule
\textbf{Winner} & Tripadvisor \\
\bottomrule
\end{tabularx}
\caption{Case Study 2: Vacation vertical. In the generated content, the winning advertiser is shown in bold, and the highlighted sentence is the span allocated to the hyperlink.}
\label{tab:case_study_vacation}
\end{table*}

\clearpage

\begin{table*}[htbp]
\centering
\small
\renewcommand{\arraystretch}{1.3}
\begin{tabularx}{\textwidth}{l X}
\toprule
\textbf{Field} & \textbf{Content} \\
\midrule
\textbf{Vertical} & Car \\
\midrule
\textbf{Query} & cars for sale near me \\
\midrule
\textbf{Reference} & Find cars for sale near you by checking local dealerships or online listings that highlight nearby inventory. Many dealers offer in-person visits, virtual tours, and flexible scheduling to help you explore options conveniently. When looking for cars for sale near me, consider visiting multiple dealers to compare available models, pricing, and financing options. Local dealers often carry a range of new and used vehicles, including popular brands like Toyota, Honda, and Ford, and can provide details about vehicle history, warranty coverage, and service records. It's also a good idea to confirm current availability and pricing directly with the dealership, as inventory and offers can change frequently. Some dealers may offer test drives, maintenance packages, or special promotions to help you make an informed decision. Whether you're looking for a new car, used car, hybrid, or electric vehicle, checking with local dealers can provide a personalized shopping experience tailored to your needs and preferences. Be sure to ask about any current incentives or financing deals that may be available to you. \\
\midrule
\textbf{LAMA} & If you're looking for cars for sale near you, checking with local dealers is the best way to find current inventory. \winnerhl{\textbf{@CARFAX@} can help you explore used vehicles available at dealers near you by providing detailed history reports, so you can research cars for sale in your area with confidence.} Local dealers typically offer in-person visits, test drives, and financing options, making it easier to compare models and secure a deal. To locate nearby dealers with suitable options, you can use your device’s map service or the search function on automotive websites. Be sure to call ahead to confirm availability and set up a visit, as popular models may sell quickly. If you're interested in a specific make or model, dealers often have multiple units on-site or can order one for you. Used cars for sale near you usually come with a range of options, including different years, trims, and mileage, so you can find a fit for your needs and budget. @Kelley Blue Book@ is another resource that can help you understand fair pricing for used vehicles at dealers near you. By combining these tools with direct contact to local dealers, you can streamline your search for the right car in your area. \\
\midrule
\textbf{Winner} & CARFAX \\
\bottomrule
\end{tabularx}
\caption{Case Study 3: Car vertical. In the generated content, the winning advertiser is shown in bold, and the highlighted sentence is the span allocated to the hyperlink.}
\label{tab:case_study_car}
\end{table*}

\begin{table*}[t]
\centering
\small
\setlength{\tabcolsep}{4pt}
\renewcommand{\arraystretch}{1.08}
\begin{tabularx}{\textwidth}{@{}r r X@{}}
\toprule
\textbf{Offset} & \textbf{$\lambda$} & \textbf{Content} \\
\midrule
$-0.2$ & $0.007$ & \blueul{Used cars }\redul{are}\blueul{ a popular choice for buyers }\redul{seeking}\blueul{ reliable transportation at }\redul{a}\blueul{ lower }\redul{cost}\blueul{, offering }\redul{a}\blueul{ range of models, years, and conditions}\redul{.}\winnerblueul{ \textbf{@CARFAX@} is a trusted resource that provides detailed insights into a used car's history, }\winnerredul{including}\winnerblueul{ accident reports, service records, and ownership changes, which can be especially helpful when you're considering a specific used car.}\blueul{ The }\redul{key}\blueul{ factors to consider include the vehicle's mileage, overall condition, and any available service history. }\redul{For}\blueul{ models typically found in used inventory}\redul{,}\blueul{ like the Honda Civic or Toyota Corolla, their reputation for reliability and fuel efficiency makes them top contenders. Buyers should also pay attention to the seller type}\redul{---}\blueul{whether it's a private seller or }\redul{a}\blueul{ certified used dealer---as this can impact the car}\redul{'s}\blueul{ condition and included services. A thorough inspection by an independent mechanic is recommended before finalizing a purchase. Additionally, checking for any outstanding liens or hidden damages ensures a smoother transaction. Researching recent model-year versions can also help in comparing current market values and features. Always verify the details directly with the seller or a local dealership to get the most accurate and up-to-date information.} \\
\midrule
$0$ & $0.074$ & \redul{If}\blueul{ you}\redul{'re looking}\blueul{ for used cars}\redul{, start by}\blueul{ evaluating models }\redul{based on their mileage, condition,}\blueul{ and pricing }\redul{history,}\blueul{ with a }\redul{focus on}\blueul{ popular options }\redul{like the Honda Civic}\blueul{, Toyota Corolla}\redul{, and}\blueul{ Hyundai Elan}\redul{tra,}\blueul{ which typically offer strong }\redul{resale}\blueul{ value }\redul{and}\blueul{ long}\redul{-term}\blueul{ reliability. Used }\redul{cars}\blueul{ from }\redul{201}\blueul{5 to }\redul{202}\blueul{0 }\redul{are}\blueul{ especially common }\redul{and}\blueul{ often provide a }\redul{good}\blueul{ balance }\redul{of modern}\blueul{ features }\redul{and affordability.}\winnerblueul{ \textbf{@CARFAX@} is a }\winnerredul{valuable resource that can}\winnerblueul{ help }\winnerredul{you}\winnerblueul{ review the }\winnerredul{history of}\winnerblueul{ used cars}\winnerredul{,}\winnerblueul{ including accident reports, service records, and ownership details, which can inform }\winnerredul{your}\winnerblueul{ decision-making.}\blueul{ When shopping, consider the type }\redul{of}\blueul{ seller}\redul{---}\blueul{dealerships often provide certified }\redul{pre-owned vehicles}\blueul{ with warranties}\redul{,}\blueul{ while }\redul{private sellers}\blueul{ may offer }\redul{lower}\blueul{ prices }\redul{but}\blueul{ less protection. Be sure }\redul{to}\blueul{ inspect the car thoroughly, }\redul{ideally}\blueul{ with a professional mechanic, and ask for a detailed }\redul{vehicle}\blueul{ history }\redul{report}\blueul{. The best deals can be found by comparing }\redul{listings}\blueul{ across trusted platforms and understanding the local market to identify fair pricing and }\redul{hidden}\blueul{ value.} \\
\midrule
$+0.2$ & $0.926$ & \redul{If you're looking for used}\blueul{ cars}\redul{, start by}\blueul{ evaluating }\redul{models based on their mileage, condition,}\blueul{ and pricing }\redul{history, with popular options including the Honda Civic, Toyota Corolla, and Hyundai Elantra, which are known}\blueul{ for }\redul{their depend}\blueul{ability }\redul{and value in the pre-owned}\blueul{ market}\redul{. Used cars typically range in price from a few thousand to over \$30}\blueul{,}\redul{0}\blueul{00}\redul{, depending}\blueul{ on }\redul{the year, trim, and features,}\blueul{ and }\redul{it's essential to review their}\blueul{ complete }\redul{vehicle history for any past accidents or maintenance records.}\winnerredul{ \textbf{@Kelley Blue Book@} is a}\winnerblueul{ helpful }\winnerredul{resource that provides insights into used cars}\winnerblueul{ for }\winnerredul{buyers, offering tools to check fair market values, set realistic asking prices, and understand current trends.}\redul{ When considering a used car, always inspect the vehicle thoroughly or have a qualified mechanic perform a pre-purchase inspection, especially for models with higher mileage.}\blueul{ @CARFAX@ }\redul{can also provide a detailed}\blueul{ look }\redul{at a used car's history, including title status and service records, so you can make a more informed decision before finalizing a purchase.} \\
\bottomrule
\end{tabularx}
\caption{Car vertical, query ``used cars'', with participating advertisers \textcolor{adred}{Kelley Blue Book (KBB)} and \textcolor{adblue}{CARFAX}. \textcolor{adred}{KBB} primarily contributes used-car pricing and market-value guidance, while \textcolor{adblue}{CARFAX} contributes vehicle-history and title/service-record evidence; the pair is mainly complementary for used-car shopping, with limited overlap as automotive information resources. \textcolor{adred}{KBB}- and \textcolor{adblue}{CARFAX}-generated tokens are underlined in their respective colors. In the generated content, the winning advertiser is shown in bold, and the highlighted sentence is the span allocated to the hyperlink.}
\label{tab:usedcars-kbb-carfax}
\end{table*}

\begin{table*}[t]
\centering
\small
\setlength{\tabcolsep}{4pt}
\renewcommand{\arraystretch}{1.08}
\begin{tabularx}{\textwidth}{@{}r r X@{}}
\toprule
\textbf{Offset} & \textbf{$\lambda$} & \textbf{Content} \\
\midrule
$-0.2$ & $0.019$ & \greenul{If you're looking for used }\redul{cars}\greenul{, start by evaluating models based on their mileage, condition, and }\redul{pricing}\greenul{ history, with popular options including the Honda Civic, Toyota Corolla, and Hyundai Elantra, which are known }\redul{for}\greenul{ their reliability }\redul{and}\greenul{ affordability. Used cars from 2015 to 2020 typically offer a good balance of modern features and }\redul{lower prices}\greenul{, and their value depends largely on maintenance records, accident history, and overall exterior and interior }\redul{condition}\greenul{.}\winnergreenul{ \textbf{@Edmunds@} can help you explore detailed listings and pricing guidance for used }\winnerredul{cars}\winnergreenul{, providing insights into fair market value }\winnerredul{and}\winnergreenul{ common features across different years and trims.}\greenul{ Be sure to check the vehicle history report for any red flags and consider scheduling a pre-purchase inspection by a certified mechanic to ensure the car you're buying is in solid shape. When shopping for used cars}\redul{,}\greenul{ compare sellers like private individuals, dealerships, and certified pre-owned listings, as each offers different levels of warranty and service support, and always test drive the vehicle to confirm it meets your expectations for performance and comfort.} \\
\midrule
$0$ & $0.546$ & \greenul{If you're looking }\redul{for used cars}\greenul{, start by }\redul{evaluating models}\greenul{ based on their mileage, condition, }\redul{and pricing}\greenul{ history, }\redul{with popular}\greenul{ options including }\redul{the Honda}\greenul{ Civic, Toyota Cor}\redul{olla}\greenul{, }\redul{and}\greenul{ Hyundai Elan}\redul{tra, which}\greenul{ are known }\redul{for their}\greenul{ reliability }\redul{and}\greenul{ affordability}\redul{.}\greenul{ Used cars from }\redul{20}\greenul{1}\redul{5 to}\greenul{ 20}\redul{20 typically}\greenul{ offer a good }\redul{balance}\greenul{ of }\redul{modern features}\greenul{ and }\redul{lower prices}\greenul{, and }\redul{their}\greenul{ value depends }\redul{on}\greenul{ factors like maintenance records}\redul{, accident}\greenul{ history, and service }\redul{history}\greenul{. }\winnerredul{\textbf{@Kelley Blue Book@} is}\winnergreenul{ a trusted }\winnerredul{resource that}\winnergreenul{ can }\winnerredul{help you understand the right price range for used cars based on their year}\winnergreenul{, }\winnerredul{model, trim, mileage}\winnergreenul{, }\winnerredul{and location.}\greenul{ @Edmunds@ }\redul{can}\greenul{ also assist }\redul{in}\greenul{ exploring }\redul{detailed}\greenul{ insights about used cars, }\redul{including common}\greenul{ mechanical }\redul{issues}\greenul{, repair costs}\redul{,}\greenul{ and }\redul{owner}\greenul{ reviews, }\redul{so you}\greenul{ can }\redul{make}\greenul{ a }\redul{more informed}\greenul{ decision}\redul{.}\greenul{ Be sure }\redul{to}\greenul{ get }\redul{a professional}\greenul{ inspection }\redul{before finalizing}\greenul{ a }\redul{purchase}\greenul{, }\redul{especially}\greenul{ for }\redul{used}\greenul{ cars }\redul{with}\greenul{ higher }\redul{mileage, to ensure}\greenul{ there are }\redul{no hidden mechanical}\greenul{ problems that could increase }\redul{your long-term}\greenul{ costs}\redul{.} \\
\midrule
$+0.2$ & $0.974$ & \redul{If you}\greenul{'re }\redul{looking for used cars,}\greenul{ start }\redul{by evaluating models}\greenul{ based }\redul{on}\greenul{ their }\redul{mileage,}\greenul{ condition}\redul{, and pricing history, with popular options}\greenul{ including }\redul{the Honda Civic}\greenul{, }\redul{Toyota Corolla, and Hyundai Elantra, which are}\greenul{ known }\redul{for their dependability and value in the}\greenul{ pre}\redul{-owned market. Used cars typically range}\greenul{ in }\redul{price from a few thousand to}\greenul{ over }\redul{\$}\greenul{3}\redul{0,000, depending on the}\greenul{ year}\redul{, trim, and features, and}\greenul{ it's essential to }\redul{review}\greenul{ a }\redul{comprehensive vehicle}\greenul{ history }\redul{report to ensure}\greenul{ there }\redul{are no hidden issues.}\winnerredul{ \textbf{@Kelley Blue Book@} is a trusted resource that provides detailed pricing guidance for used cars, offering insights into fair market values and}\winnergreenul{ what }\winnerredul{to expect when you're ready to}\winnergreenul{ negotiate}\winnerredul{.}\redul{ Be sure to inspect the car thoroughly or}\greenul{ have }\redul{a qualified mechanic do so, and consider the seller type---whether a private individual or a certified dealer---as this can impact}\greenul{ the }\redul{overall reliability and warranty options available with the used cars you're considering.} \\
\bottomrule
\end{tabularx}
\caption{Car vertical, query ``used cars'', with participating advertisers \textcolor{adred}{Kelley Blue Book (KBB)} and \textcolor{adgreen}{Edmunds}. Both advertisers can answer used-car research and pricing-comparison intent, so the pair is more competitive than \textcolor{adred}{KBB}--\textcolor{adblue}{CARFAX}; they are partially complementary because \textcolor{adred}{KBB} emphasizes valuation while \textcolor{adgreen}{Edmunds} emphasizes reviews, model research, and shopping guidance. \textcolor{adred}{KBB}- and \textcolor{adgreen}{Edmunds}-generated tokens are underlined in their respective colors. In the generated content, the winning advertiser is shown in bold, and the highlighted sentence is the span allocated to the hyperlink.}
\label{tab:usedcars-kbb-edmunds}
\end{table*}

\begin{table*}[t]
    \centering
    \caption{Latency and seconds-per-token ratios (LAMA/Reference) across verticals. These ratios are close to the theoretical prediction ($|\mathcal{N}|+1$)}.
    \label{tab:latency-ratio}
    \begin{tabular}{lcccc}
        \toprule
        \textbf{Vertical}
        & \textbf{Latency Mean}
        & \textbf{p50}
        & \textbf{p90}
        & \textbf{Sec/Token} \\
        \midrule
        Workout  & $3.87\times$ & $3.90\times$ & $4.31\times$ & $3.51\times$ \\
        Vacation & $3.53\times$ & $3.47\times$ & $3.90\times$ & $3.53\times$ \\
        Car      & $3.94\times$ & $3.97\times$ & $4.56\times$ & $3.57\times$ \\
        \bottomrule
    \end{tabular}
\end{table*}

\begin{table*}[t]
    \centering
    \caption{Impression-value prediction and calibration performance across verticals. In our experiments, impression values lie in $[0,1]$ and are defined as click-through rates multiplied by a unit cost per click. The current experimental setup (Qwen3-14B + LoRA fine-tuning) already yields reasonably good recovery accuracy.}
    \label{tab:ctr-calibration}
    \begin{tabular}{lccc}
        \toprule
        \textbf{Vertical}
        & \textbf{MAE}
        & \textbf{Rel. MAE}
        & \textbf{ECE} \\
        \midrule
        Workout  & 0.052 & 7.1\% & 0.0149 \\
        Vacation & 0.065 & 8.5\% & 0.0300 \\
        Car      & 0.056 & 7.4\% & 0.0375 \\
        \midrule
        Macro Avg. & 0.058 & 7.7\% & 0.0275 \\
        \bottomrule
    \end{tabular}
\end{table*}

\clearpage

\begin{figure*}[t]
  \centering
  \includegraphics[width=0.5\textwidth]{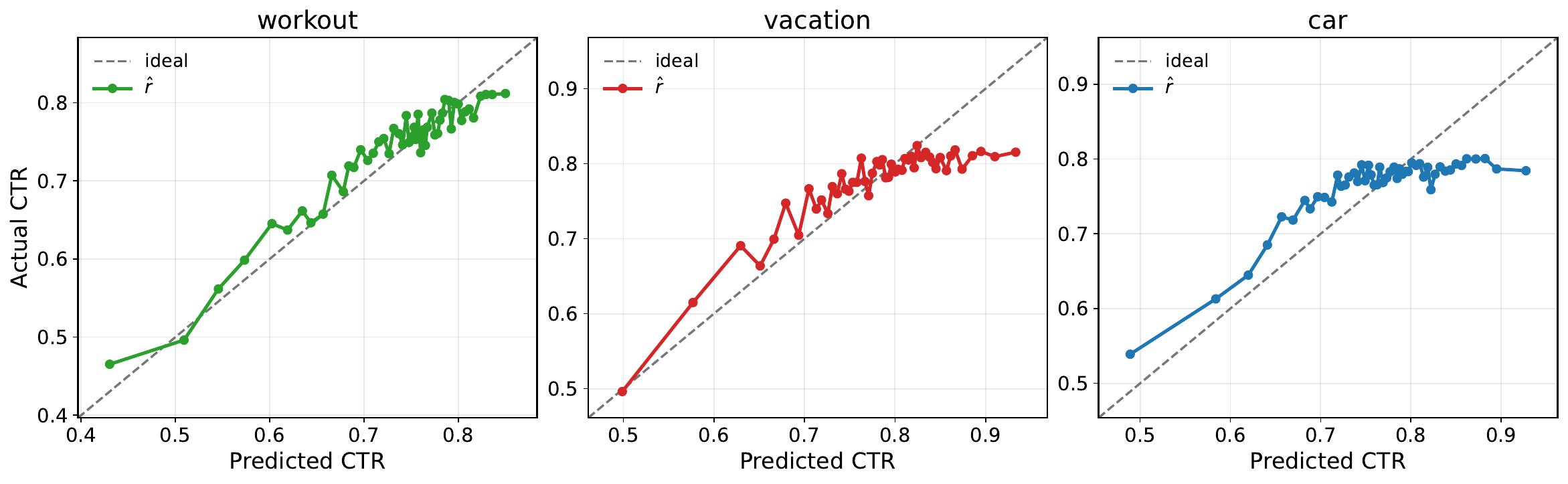}
  \caption{
    Quantile-binned calibration curves for the value model $\hat r$ across the three verticals.
    Each panel corresponds to one vertical, and each point represents one of 50 equal-frequency bins sorted by predicted impression value.
    The x-axis reports the mean predicted impression value within the bin, while the y-axis reports the mean ground-truth impression value for the same examples.
  }
  \Description{Calibration curves comparing mean predicted and ground-truth click-through rates across the car, workout, and vacation verticals. The current experimental setup demonstrates generally acceptable calibration performance, although it tends to overestimate CTRs in the high-CTR range.}
  \label{fig:value-model-calibration}
\end{figure*}

\begin{figure*}[t]
  \centering
  \includegraphics[width=0.5\textwidth]{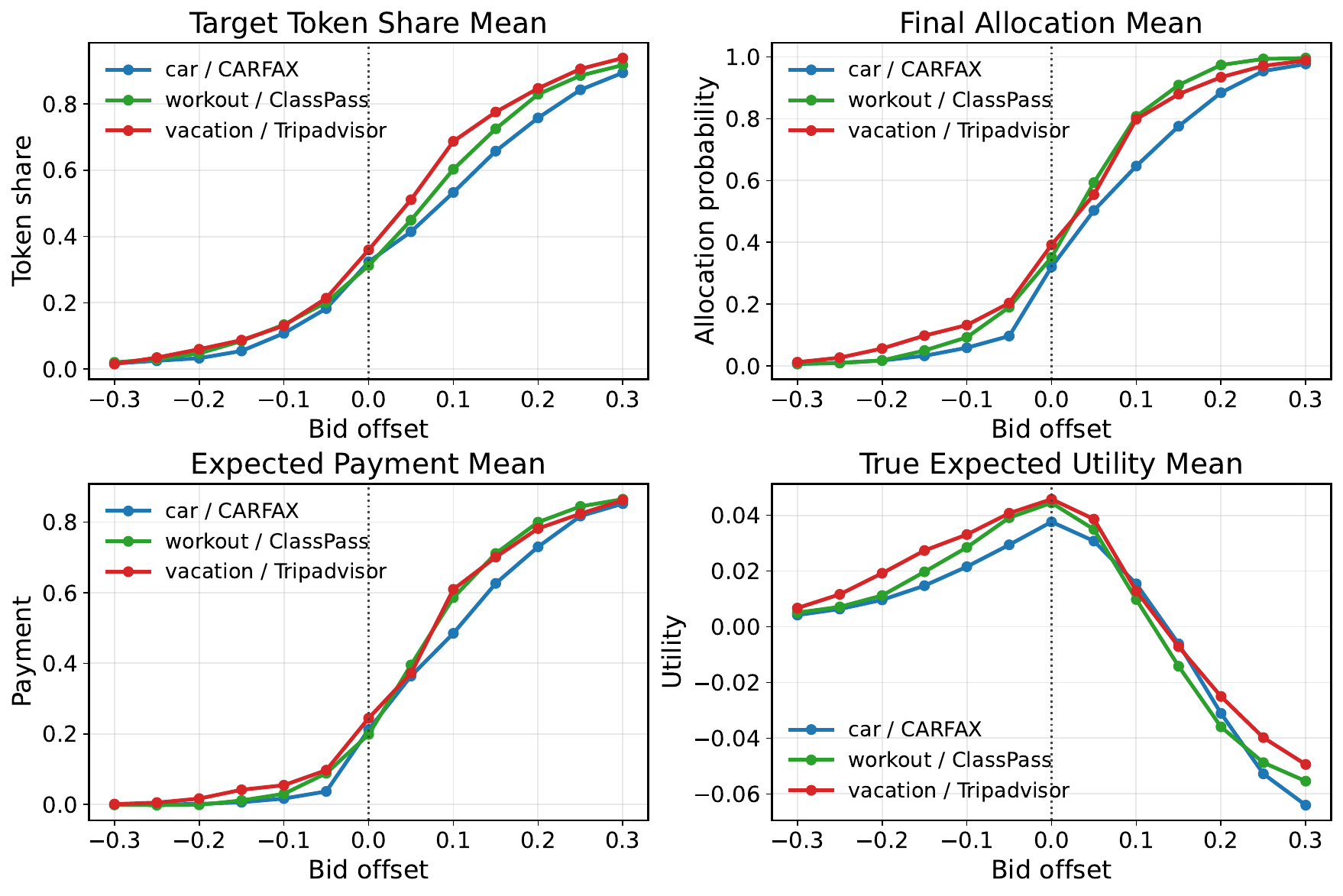}
  \caption{
    Bid-offset sweep for the LAMA token auction across the three verticals.
    For each vertical, we select one target advertiser and add a constant offset to its truthful value $v(q)$ while leaving the other advertisers' bids unchanged.
    The target advertisers are CARFAX for car, ClassPass for workout, and Tripadvisor for vacation.
    The sweep ranges from $-0.3$ to $+0.3$, while the advantage model settings remain fixed to the vertical-specific configuration.
    The four panels report the mean target token share, final allocation probability, expected payment, and true expected utility.
    Across all three verticals, increasing the bid offset monotonically increases token share, allocation, and payment.
    In contrast, true utility peaks at the truthful bid for all three verticals, then declines under overbidding as higher payments offset the additional allocation.
    This pattern provides qualitative evidence that the mechanism rewards truthful bidding while still allowing bids to control token-level allocation.
  }
  \Description{Bid-offset sweeps showing target token share, allocation probability, expected payment, and true expected utility for one advertiser in each vertical.}
  \label{fig:bid-offset-sweep}
\end{figure*}

%% file: classic_ad_auction_papers.bib
@article{Edelman_2007,
  title     = {Internet Advertising and the Generalized Second-Price Auction: Selling Billions of Dollars Worth of Keywords},
  volume    = {97},
  ISSN      = {0002-8282},
  url       = {http://dx.doi.org/10.1257/aer.97.1.242},
  DOI       = {10.1257/aer.97.1.242},
  number    = {1},
  journal   = {American Economic Review},
  publisher = {American Economic Association},
  author    = {Edelman, Benjamin and Ostrovsky, Michael and Schwarz, Michael},
  year      = {2007},
  month     = feb,
  pages     = {242--259}
}

@article{Varian_2007,
  title     = {Position Auctions},
  volume    = {25},
  ISSN      = {0167-7187},
  url       = {http://dx.doi.org/10.1016/j.ijindorg.2006.10.002},
  DOI       = {10.1016/j.ijindorg.2006.10.002},
  number    = {6},
  journal   = {International Journal of Industrial Organization},
  publisher = {Elsevier BV},
  author    = {Varian, Hal R.},
  year      = {2007},
  month     = dec,
  pages     = {1163--1178}
}

@inproceedings{DBLP:conf/wine/AggarwalFMP08,
  author       = {Gagan Aggarwal and
                  Jon Feldman and
                  S. Muthukrishnan and
                  Martin P{\'{a}}l},
  editor       = {Christos H. Papadimitriou and
                  Shuzhong Zhang},
  title        = {Sponsored Search Auctions with Markovian Users},
  booktitle    = {Internet and Network Economics, 4th International Workshop, {WINE}
                  2008, Shanghai, China, December 17-20, 2008. Proceedings},
  series       = {Lecture Notes in Computer Science},
  pages        = {621--628},
  publisher    = {Springer},
  year         = {2008},
  url          = {https://doi.org/10.1007/978-3-540-92185-1\_68},
  doi          = {10.1007/978-3-540-92185-1\_68},
  bibsource    = {dblp computer science bibliography, https://dblp.org}
}

@inproceedings{DBLP:conf/ecai/0001RSV16,
  author       = {Nicola Gatti and
                  Marco Rocco and
                  Paolo Serafino and
                  Carmine Ventre},
  editor       = {Gal A. Kaminka and
                  Maria Fox and
                  Paolo Bouquet and
                  Eyke H{\"{u}}llermeier and
                  Virginia Dignum and
                  Frank Dignum and
                  Frank van Harmelen},
  title        = {Towards Better Models of Externalities in Sponsored Search Auctions},
  booktitle    = {{ECAI} 2016 - 22nd European Conference on Artificial Intelligence,
                  29 August-2 September 2016, The Hague, The Netherlands - Including
                  Prestigious Applications of Artificial Intelligence {(PAIS} 2016)},
  series       = {Frontiers in Artificial Intelligence and Applications},
  pages        = {1167--1175},
  publisher    = {{IOS} Press},
  year         = {2016},
  url          = {https://doi.org/10.3233/978-1-61499-672-9-1167},
  doi          = {10.3233/978-1-61499-672-9-1167},
  bibsource    = {dblp computer science bibliography, https://dblp.org}
}

@article{DBLP:journals/tcs/0001RSV18,
  author       = {Nicola Gatti and
                  Marco Rocco and
                  Paolo Serafino and
                  Carmine Ventre},
  title        = {Towards better models of externalities in sponsored search auctions},
  journal      = {Theor. Comput. Sci.},
  volume       = {745},
  pages        = {150--162},
  year         = {2018},
  url          = {https://doi.org/10.1016/j.tcs.2018.06.011},
  doi          = {10.1016/J.TCS.2018.06.011},
  bibsource    = {dblp computer science bibliography, https://dblp.org}
}

@inproceedings{DBLP:conf/wine/CavalloW14,
  author       = {Ruggiero Cavallo and
                  Christopher A. Wilkens},
  editor       = {Tie{-}Yan Liu and
                  Qi Qi and
                  Yinyu Ye},
  title        = {{GSP} with General Independent Click-through-Rates},
  booktitle    = {Web and Internet Economics - 10th International Conference, {WINE}
                  2014, Beijing, China, December 14-17, 2014. Proceedings},
  series       = {Lecture Notes in Computer Science},
  pages        = {400--416},
  publisher    = {Springer},
  year         = {2014},
  url          = {https://doi.org/10.1007/978-3-319-13129-0\_32},
  doi          = {10.1007/978-3-319-13129-0\_32},
  bibsource    = {dblp computer science bibliography, https://dblp.org}
}

@article{Carrion_2023,
  title     = {Blending Advertising with Organic Content in E-commerce via Virtual Bids},
  volume    = {37},
  ISSN      = {2159-5399},
  url       = {http://dx.doi.org/10.1609/aaai.v37i13.26835},
  DOI       = {10.1609/aaai.v37i13.26835},
  number    = {13},
  journal   = {Proceedings of the AAAI Conference on Artificial Intelligence},
  publisher = {Association for the Advancement of Artificial Intelligence (AAAI)},
  author    = {Carrion, Carlos and Wang, Zenan and Nair, Harikesh and Luo, Xianghong and Lei, Yulin and Gu, Peiqin and Lin, Xiliang and Chen, Wenlong and Jin, Junsheng and Zhu, Fanan and Peng, Changping and Bao, Yongjun and Lin, Zhangang and Yan, Weipeng and Shao, Jingping},
  year      = {2023},
  month     = jun,
  pages     = {15476--15484}
}

@inproceedings{DBLP:conf/www/GhoshM08,
  author       = {Arpita Ghosh and
                  Mohammad Mahdian},
  editor       = {Jinpeng Huai and
                  Robin Chen and
                  Hsiao{-}Wuen Hon and
                  Yunhao Liu and
                  Wei{-}Ying Ma and
                  Andrew Tomkins and
                  Xiaodong Zhang},
  title        = {Externalities in online advertising},
  booktitle    = {Proceedings of the 17th International Conference on World Wide Web,
                  {WWW} 2008, Beijing, China, April 21-25, 2008},
  pages        = {161--168},
  publisher    = {{ACM}},
  year         = {2008},
  url          = {https://doi.org/10.1145/1367497.1367520},
  doi          = {10.1145/1367497.1367520},
  bibsource    = {dblp computer science bibliography, https://dblp.org}
}

@inproceedings{DBLP:conf/wine/GiotisK08,
  author       = {Ioannis Giotis and
                  Anna R. Karlin},
  editor       = {Christos H. Papadimitriou and
                  Shuzhong Zhang},
  title        = {On the Equilibria and Efficiency of the {GSP} Mechanism in Keyword
                  Auctions with Externalities},
  booktitle    = {Internet and Network Economics, 4th International Workshop, {WINE}
                  2008, Shanghai, China, December 17-20, 2008. Proceedings},
  series       = {Lecture Notes in Computer Science},
  pages        = {629--638},
  publisher    = {Springer},
  year         = {2008},
  url          = {https://doi.org/10.1007/978-3-540-92185-1\_69},
  doi          = {10.1007/978-3-540-92185-1\_69},
  bibsource    = {dblp computer science bibliography, https://dblp.org}
}

@article{Li_2023,
  title     = {Optimally integrating ad auction into e-commerce platforms},
  volume    = {976},
  ISSN      = {0304-3975},
  url       = {http://dx.doi.org/10.1016/j.tcs.2023.114141},
  DOI       = {10.1016/j.tcs.2023.114141},
  journal   = {Theoretical Computer Science},
  publisher = {Elsevier BV},
  author    = {Li, Weian and Qi, Qi and Wang, Changjun and Yu, Changyuan},
  year      = {2023},
  month     = oct,
  pages     = {114141}
}

@article{An_2025,
  title     = {Merging Mechanisms for Ads and Organic Items in E-commerce Platforms},
  volume    = {39},
  ISSN      = {2159-5399},
  url       = {http://dx.doi.org/10.1609/aaai.v39i13.33479},
  DOI       = {10.1609/aaai.v39i13.33479},
  number    = {13},
  journal   = {Proceedings of the AAAI Conference on Artificial Intelligence},
  publisher = {Association for the Advancement of Artificial Intelligence (AAAI)},
  author    = {An, Nan and Li, Weian and Qi, Qi and Zhang, Liang},
  year      = {2025},
  month     = apr,
  pages     = {13547--13554}
}


%% file: llm_advertising_papers.bib
@inproceedings{DBLP:conf/www/DuttingMLXZ24,
  author       = {Paul D{\"{u}}tting and
                  Vahab Mirrokni and
                  Renato Paes Leme and
                  Haifeng Xu and
                  Song Zuo},
  editor       = {Tat{-}Seng Chua and
                  Chong{-}Wah Ngo and
                  Ravi Kumar and
                  Hady W. Lauw and
                  Roy Ka{-}Wei Lee},
  title        = {Mechanism Design for Large Language Models},
  booktitle    = {Proceedings of the {ACM} on Web Conference 2024, {WWW} 2024, Singapore,
                  May 13-17, 2024},
  pages        = {144--155},
  publisher    = {{ACM}},
  year         = {2024},
  url          = {https://doi.org/10.1145/3589334.3645511},
  doi          = {10.1145/3589334.3645511},
  bibsource    = {dblp computer science bibliography, https://dblp.org}
}

@inproceedings{DBLP:conf/kdd/Dubey0KM024,
  author       = {Avinava Dubey and
                  Zhe Feng and
                  Rahul Kidambi and
                  Aranyak Mehta and
                  Di Wang},
  editor       = {Ricardo Baeza{-}Yates and
                  Francesco Bonchi},
  title        = {Auctions with {LLM} Summaries},
  booktitle    = {Proceedings of the 30th {ACM} {SIGKDD} Conference on Knowledge Discovery
                  and Data Mining, {KDD} 2024, Barcelona, Spain, August 25-29, 2024},
  pages        = {713--722},
  publisher    = {{ACM}},
  year         = {2024},
  url          = {https://doi.org/10.1145/3637528.3672022},
  doi          = {10.1145/3637528.3672022},
  bibsource    = {dblp computer science bibliography, https://dblp.org}
}

@article{DBLP:journals/corr/abs-2405-05905,
  author       = {Ermis Soumalias and
                  Michael J. Curry and
                  Sven Seuken},
  title        = {Truthful Aggregation of LLMs with an Application to Online Advertising},
  journal      = {CoRR},
  volume       = {abs/2405.05905},
  year         = {2024},
  url          = {https://doi.org/10.48550/arXiv.2405.05905},
  doi          = {10.48550/ARXIV.2405.05905},
  eprinttype   = {arXiv},
  eprint       = {2405.05905},
  bibsource    = {dblp computer science bibliography, https://dblp.org}
}

@misc{xu2026adinsertionllmgeneratedresponses,
  title         = {Ad Insertion in LLM-Generated Responses},
  author        = {Shengwei Xu and Zhaohua Chen and Xiaotie Deng and Zhiyi Huang and Grant Schoenebeck},
  year          = {2026},
  eprint        = {2601.19435},
  archivePrefix = {arXiv},
  primaryClass  = {cs.GT},
  url           = {https://arxiv.org/abs/2601.19435}
}

@misc{sun2026lerallmenhancedragad,
  title         = {LERA: LLM-Enhanced RAG for Ad Auction in Generative Chatbots},
  author        = {Haoran Sun and Xinrui Song and Xinyu Zhang and Zhaohua Chen and Xu Chu and Zhilin Zhang and Chuan Yu and Jian Xu and Bo Zheng and Xiaotie Deng},
  year          = {2026},
  eprint        = {2605.16474},
  archivePrefix = {arXiv},
  primaryClass  = {cs.IR},
  url           = {https://arxiv.org/abs/2605.16474}
}

@misc{yun2026llmadvertisementbasedneuron,
  title         = {LLM Advertisement based on Neuron Auctions},
  author        = {Peiran Yun and Wenxin Xu and Jiayuan Liu and Yihang Zhang and Liang Zeng and Lingkai Kong and Tonghan Wang},
  year          = {2026},
  eprint        = {2605.08326},
  archivePrefix = {arXiv},
  primaryClass  = {cs.LG},
  url           = {https://arxiv.org/abs/2605.08326}
}

@misc{hu2026gembenchbenchmarkadinjectedresponse,
  title         = {GEM-Bench: A Benchmark for Ad-Injected Response Generation within Generative Engine Marketing},
  author        = {Silan Hu and Shiqi Zhang and Yimin Shi and Xiaokui Xiao},
  year          = {2026},
  eprint        = {2509.14221},
  archivePrefix = {arXiv},
  primaryClass  = {cs.IR},
  url           = {https://arxiv.org/abs/2509.14221}
}

@inproceedings{NEURIPS2024_20dcab0f,
  author    = {Hajiaghayi, MohammadTaghi and Lahaie, S{\'{e}}bastien and Rezaei, Keivan and Shin, Suho},
  booktitle = {Advances in Neural Information Processing Systems},
  doi       = {10.52202/079017-0585},
  editor    = {A. Globerson and L. Mackey and D. Belgrave and A. Fan and U. Paquet and J. Tomczak and C. Zhang},
  pages     = {18445--18480},
  publisher = {Curran Associates, Inc.},
  title     = {Ad Auctions for LLMs via Retrieval Augmented Generation},
  url       = {https://proceedings.neurips.cc/paper_files/paper/2024/file/20dcab0f14046a5c6b02b61da9f13229-Paper-Conference.pdf},
  volume    = {37},
  year      = {2024}
}

@misc{balseiro2025positionauctionsaigeneratedcontent,
  title         = {Position Auctions in AI-Generated Content},
  author        = {Santiago Balseiro and Kshipra Bhawalkar and Yuan Deng and Zhe Feng and Jieming Mao and Aranyak Mehta and Vahab Mirrokni and Renato Paes Leme and Di Wang and Song Zuo},
  year          = {2025},
  eprint        = {2506.03309},
  archivePrefix = {arXiv},
  primaryClass  = {cs.GT},
  url           = {https://arxiv.org/abs/2506.03309}
}

@misc{han2026mechanismdesignqualitypreservingllmadvertising,
  title         = {Mechanism Design for Quality-Preserving LLM Advertising},
  author        = {Jiale Han and Xiaowu Dai},
  year          = {2026},
  eprint        = {2605.10964},
  archivePrefix = {arXiv},
  primaryClass  = {cs.GT},
  url           = {https://arxiv.org/abs/2605.10964}
}

@misc{zhao2026llmauctiongenerativeauctionllmnative,
  title         = {LLM-Auction: Generative Auction towards LLM-Native Advertising},
  author        = {Chujie Zhao and Qun Hu and Shiping Song and Dagui Chen and Han Zhu and Jian Xu and Bo Zheng},
  year          = {2026},
  eprint        = {2512.10551},
  archivePrefix = {arXiv},
  primaryClass  = {cs.GT},
  url           = {https://arxiv.org/abs/2512.10551}
}

@inproceedings{schmidt2024detecting,
  title     = {Detecting Generated Native Ads in Conversational Search},
  author    = {Sebastian Schmidt and Ines Zelch and Janek Bevendorff and Benno Stein and Matthias Hagen and Martin Potthast},
  booktitle = {Companion Proceedings of the ACM Web Conference 2024},
  pages     = {722--725},
  publisher = {{ACM}},
  year      = {2024},
  doi       = {10.1145/3589335.3651489},
  url       = {https://doi.org/10.1145/3589335.3651489}
}

@article{qwen3embedding,
  title   = {Qwen3 Embedding: Advancing Text Embedding and Reranking Through Foundation Models},
  author  = {Zhang, Yanzhao and Li, Mingxin and Long, Dingkun and Zhang, Xin and Lin, Huan and Yang, Baosong and Xie, Pengjun and Yang, An and Liu, Dayiheng and Lin, Junyang and Huang, Fei and Zhou, Jingren},
  journal = {arXiv preprint arXiv:2506.05176},
  year    = {2025},
  doi     = {10.48550/arXiv.2506.05176}
}

@misc{lv12_esci_msmarco_minilm_l12_v2,
  title        = {{esci-ms-marco-MiniLM-L-12-v2}},
  author       = {{lv12}},
  year         = {2024},
  publisher    = {Hugging Face},
  howpublished = {\url{https://huggingface.co/lv12/esci-ms-marco-MiniLM-L-12-v2}},
  note         = {Cross-encoder fine-tuned on the Amazon ESCI dataset; based on cross-encoder/ms-marco-MiniLM-L12-v2}
}

@article{yang2025qwen3,
  title={Qwen3 technical report},
  author={Yang, An and Li, Anfeng and Yang, Baosong and Zhang, Beichen and Hui, Binyuan and Zheng, Bo and Yu, Bowen and Gao, Chang and Huang, Chengen and Lv, Chenxu and others},
  journal={arXiv preprint arXiv:2505.09388},
  year={2025}
}

@article{hu2021lora,
  title={Lora: Low-rank adaptation of large language models},
  author={Hu, Edward J and Shen, Yelong and Wallis, Phillip and Allen-Zhu, Zeyuan and Li, Yuanzhi and Wang, Shean and Wang, Lu and Chen, Weizhu},
  journal={arXiv preprint arXiv:2106.09685},
  year={2021}
}

@article{bradley1952rank,
  title={Rank analysis of incomplete block designs: I. the method of paired comparisons},
  author={Bradley, Ralph Allan and Terry, Milton E},
  journal={Biometrika},
  volume={39},
  number={3/4},
  pages={324--345},
  year={1952},
  publisher={JSTOR}
}

@article{ouyang2022training,
  title={Training language models to follow instructions with human feedback},
  author={Ouyang, Long and Wu, Jeffrey and Jiang, Xu and Almeida, Diogo and Wainwright, Carroll and Mishkin, Pamela and Zhang, Chong and Agarwal, Sandhini and Slama, Katarina and Ray, Alex and others},
  journal={Advances in neural information processing systems},
  volume={35},
  pages={27730--27744},
  year={2022}
}

@inproceedings{he2021unified,
  title={A unified solution to constrained bidding in online display advertising},
  author={He, Yue and Chen, Xiujun and Wu, Di and Pan, Junwei and Tan, Qing and Yu, Chuan and Xu, Jian and Zhu, Xiaoqiang},
  booktitle={Proceedings of the 27th ACM SIGKDD Conference on Knowledge Discovery \& Data Mining},
  pages={2993--3001},
  year={2021}
}

@inproceedings{aggarwal2019autobidding,
  title={Autobidding with constraints},
  author={Aggarwal, Gagan and Badanidiyuru, Ashwinkumar and Mehta, Aranyak},
  booktitle={International Conference on Web and Internet Economics},
  pages={17--30},
  year={2019},
  organization={Springer}
}

@book{mcluhan1964understanding,
  title={Understanding media},
  author={McLuhan, Marshall},
  volume={2002},
  year={1964},
  publisher={Routledge London}
}


%% file: sample-base.bib
@String{Computing = "Computing" }

@String{Computer = "{IEEE} Computer" }

@String{Springer = "Springer-Verlag" }

@ArtifactSoftware{R,
    title = {R: A Language and Environment for Statistical Computing},
    author = {{R Core Team}},
    organization = {R Foundation for Statistical Computing},
    address = {Vienna, Austria},
    year = {2019},
    url = {https://www.R-project.org/},
}
